\documentclass[article,11pt]{article}

\usepackage[T1]{fontenc}
\usepackage{float,graphicx,verbatim,fullpage,hyperref,amssymb,amsmath,amsthm,enumerate,multicol, xspace,xcolor,mathtools,thm-restate}
\usepackage[margin=1in]{geometry}
\usepackage{algorithm, algorithmic, cite}
\newcommand*{\email}[1]{\texttt{#1}}

\def\final{0}  
\def\iflong{\iffalse}
\ifnum\final=0  
\newcommand{\anote}[1]{{\color{red}[{\tiny Ali: \bf #1}]\marginpar{\color{red}*}}}
\newcommand{\knote}[1]{{\color{red}[{\tiny Karthik: \bf #1}]\marginpar{\color{red}*}}}
\newcommand{\cnote}[1]{{\color{red}[{\tiny Charlie: \bf #1}]\marginpar{\color{red}*}}}
\newcommand{\todo}[1]{{\color{red}[{\tiny TODO: \bf #1}]\marginpar{\color{red}*}}}
\else 
\newcommand{\anote}[1]{}
\newcommand{\knote}[1]{}
\newcommand{\cnote}[1]{}
\newcommand{\todo}[1]{}
\fi  

\newtheorem{theorem}{Theorem}[section]
\newtheorem{prop}[theorem]{Proposition}
\newtheorem{lemma}[theorem]{Lemma}
\newtheorem{claim}[theorem]{Claim}
\newtheorem{corollary}[theorem]{Corollary}

\theoremstyle{definition}
\newtheorem{definition}[theorem]{Definition}
\newtheorem{remark}[theorem]{Remark}

\def\R{\mathbb{R}}

\def\Z{\mathbb{Z}}

\makeatletter
\newcommand{\cut}[1][\@nil]{%
	\def\tmp{#1}%
	\ifx\tmp\@nnil%
		\textsc{max-cut}\xspace%
	\else%
		\textsc{max-$#1$-cut}\xspace%
	\fi%
}
\makeatother

\newcommand{\threecut}{\cut[3]}
\newcommand{\kcut}{\cut[k]}

\begin{document}
\title{Improving the smoothed complexity of FLIP for max cut problems}
\author{
Ali Bibak\thanks{University of Illinois, Urbana-Champaign, Email: \email{\{bibakse2,karthe\}@illinois.edu}.}
\and
Charles Carlson\thanks{University of Colorado, Boulder, Email: \email{chca0914@colorado.edu}.}
\and
Karthekeyan Chandrasekaran\footnotemark[1]
}
\date{}
\maketitle

\begin{abstract}
Finding locally optimal solutions for \cut and \kcut are well-known PLS-complete problems. An instinctive approach to finding such a locally optimum solution is the FLIP method. 
Even though FLIP requires exponential time in worst-case instances, it tends to terminate quickly in practical instances. 
To explain this discrepancy, the run-time of FLIP has been studied in the smoothed complexity framework. 
Etscheid and R\"{o}glin \cite{ER17} showed that the smoothed complexity of FLIP for \cut in arbitrary graphs is quasi-polynomial. Angel, Bubeck, Peres and Wei \cite{ABPW17} showed that the smoothed complexity of FLIP for \cut in complete graphs is $O(\phi^5n^{15.1})$, where $\phi$ is an upper bound on the random edge-weight density and $n$ is the number of
vertices in the input graph. 

While Angel et al.'s result showed the first polynomial smoothed complexity, they also conjectured that their run-time bound is far from optimal. In this work, we make substantial progress towards improving the run-time bound. We prove that the smoothed complexity of FLIP in complete graphs is $O(\phi n^{7.83})$. 
Our results are based on a carefully chosen matrix whose rank captures the run-time of the method along with improved rank bounds for this matrix and an improved union bound based on this matrix. In addition, our techniques provide a general framework for analyzing FLIP in the smoothed framework.
We illustrate this general framework by showing that the smoothed complexity of FLIP for \threecut in complete graphs is polynomial and for \kcut in arbitrary graphs is quasi-polynomial. We believe that our techniques should also be of interest towards addressing the smoothed complexity of FLIP for \kcut in complete graphs for larger constants $k$. 

\end{abstract}

\section{Introduction}
A $k$-cut in a graph is a partition of the vertex set into $k$ parts. 
Given an edge-weighted graph and a $k$-cut, the value of the cut is the total weight of the edges crossing the partition. In the max-$k$-cut problem, denoted \kcut, we are given a graph with edge weights and the goal is to find a $k$-cut with maximum value. For convenience, we will denote \cut[2] as \cut. 
\cut is a well-known NP-hard problem whose study brought forth new algorithmic techniques. 
In this work, we analyze the run-time of a local algorithm for \cut and more generally, for \kcut. 

A $k$-cut is said to be a \emph{local max-$k$-cut} if the cut value cannot be improved by changing the part of any \emph{single} vertex. 
We recall that a local max-$k$-cut is in fact a $(1-1/k)$-approximate max-$k$-cut \cite{KT06}. 
The computation of a local max-$k$-cut is of interest in game theory as it is a Nash-equilibrium in the party affiliation game \cite{FPT04,BCK10}: consider $n$ players with certain weights between pairs of players and they would like to form $k$ teams. The payoff for a player is the total weight of the edges between her and the players in her $k-1$ opposing teams. A local max-$k$-cut is a Nash equilibrium for this game. Sch\"{a}ffer and Yannakakis showed that computing a local max-cut is likely to be intractable. In particular, they showed that it is PLS-complete (where PLS abbreviates Polynomial-Time Local Search) \cite{SY91}. 

A natural algorithm to find a local max-$k$-cut is to start from an arbitrary $k$-cut and repeatedly perform local improving moves as long as possible. This is known as the FLIP method: it starts from an arbitrary $k$-cut and repeatedly increases the weight of the cut by moving a vertex from its current part to one of the other $k-1$ parts as long as such an improvement is possible. We note that an implementation of the FLIP method should specify how to choose (1) the initial $k$-cut, (2) the vertex to move in each iteration (if there is more than one vertex whose movement improves the cut value) and (3) the part to which the chosen vertex should be moved (if there is more than one choice that improves the cut value). The complexity of the FLIP method is the number of moves required for any implementation of the FLIP method to terminate.

The FLIP method corresponds to a natural player dynamics in the party affiliation game and hence its convergence time has been of much interest. Unfortunately, there are instances for which the FLIP method needs an exponential sequence of moves to converge even for \cut \cite{SY91, AT18}. Yet, empirical evidence suggests that the FLIP method is very fast in real-world instances of \cut \cite{JPY88}. 

The smoothed complexity framework, introduced by Spielman and Teng \cite{ST04}, is well-suited to explain the discrepancy in the performance of an algorithm between worst-case and practical instances. 
In the smoothed complexity framework for \kcut, we are given a graph $G=(V,E)$ on $n$ vertices along with a distribution $f_e:[-1,1]\rightarrow [0,\phi]$ according to which the weight $X_e$ of edge $e\in E$ is chosen. We note that the edge weights are independently distributed. Here, the parameter $\phi$ determines the amount of random noise in the instance: if $\phi\rightarrow \infty$, then the instance is essentially a worst-case instance, whereas finite $\phi$ amounts to some randomness in the instance. The restriction of the edge weights to the interval $[-1,1]$ is without loss of generality as arbitrary bounded weights can be scaled. The goal is to determine the run-time of the FLIP method in expectation (or with high probability) over the random choice of edge weights. 

We distinguish between the two models under which the smoothed complexity of the FLIP method for \cut has been studied in the literature. 
In the case of smoothed complexity for \emph{arbitrary graphs}, noise is added only to the existing edges of the given graph. In the case of smoothed complexity in \emph{complete graphs}, noise is added to all vertex pairs including the non-edges of the given graph (by treating them as zero weight edges). We emphasize that this kind of subtlety between arbitrary graphs and complete graphs while studying the smoothed complexity has been prevalent---indeed, Spielman and Teng's work analyzed the smoothed complexity of the simplex method when noise is added to every entry of the constraint matrix including the zero entries; determining the smoothed complexity of the simplex method when noise is added only to the non-zero entries of the constraint matrix still remains as an important open problem.

All previous works \cite{ET11, ER15, ER17, ABPW17} have studied the complexity of the FLIP method only for \cut while its smoothed complexity for \kcut for $k\ge 3$ has not been considered in the literature. 
We now mention the results relevant to this work. Etscheid and R\"{o}glin \cite{ER17} showed that any implementation of the FLIP method for \cut in arbitrary graphs terminates using at most $\phi n^{O(\log{n})}$ moves with high probability, i.e., the smoothed complexity when a small amount of noise is added to every edge of the given graph is quasi-polynomial. Subsequently, Angel, Bubeck, Peres and Wei \cite{ABPW17} showed that any implementation of the FLIP method for \cut in complete graphs terminates using at most $O(\phi^5n^{15.1})$ moves with high probability, i.e., the smoothed complexity when a small amount of noise is added to every vertex pair is polynomial. 

\subsection{Our results}
Motivated by empirical evidence, Angel, Bubeck, Peres and Wei conjectured that the dependence on $n$ should be quasi-linear and therefore, raised the question of improving the run-time analysis. In this work, we address this question by improving the run-time analysis of the FLIP method for \cut. 

\begin{restatable}{theorem}{thmMaxCutPoly}\label{theorem:max-cut-poly}
Let $G=(V,E)$ be the complete graph on $n$ vertices, and suppose that the edge weights $(X_e)_{e\in E}$ are independent random variables chosen according to the probability density function $f_e:[-1,1]\rightarrow [0,\phi]$ for some $\phi>0$.
For every constant $\eta > 0$, with high probability every implementation of the FLIP method for \cut terminates in at most
$1580 \phi n^{(2+\sqrt{2})(\sqrt{2}+\eta)}=O(\phi n^{7.829+3.414\eta})$ steps. 
\end{restatable}

In particular, our theorem implies that the FLIP method for \cut terminates in at most $O(\phi n^{7.83})$ moves. This run-time bound improves the dependence on both the max-density $\phi$ as well as the number of vertices $n$ in comparison to that of Angel, Bubeck, Peres and Wei \cite{ABPW17}. 
The outline of our analysis technique follows the recipe introduced by Etscheid and R\"{o}glin \cite{ER17} and followed by Angel et al. \cite{ABPW17}: we associate a suitable matrix with the sequence of FLIP moves; 
we show a union bound that a long sequence of moves followed by the FLIP method will lead to a large improvement in the cut value with high probability provided that the rank of this associated matrix is large; 
next, we show that the rank of this matrix is indeed large. Our techniques differ from those of Angel et al. along three fronts: (1) the matrix that we associate with the sequence of moves is different from the one used by Angel et al. but is closer to the one used by Etscheid and R\"{o}glin, (2) the union bound that we show is much stronger than the one by Angel et al., and (3) the rank lower bound that we show is much stronger than the one shown by Etscheid and R\"{o}glin. \\

Next, we turn to the smoothed complexity of the FLIP method for \kcut for $k\ge 3$. 
As mentioned above, all known results on the smoothed analysis of the FLIP method address only the case of $k=2$.  
However, the convergence of this method for \kcut when $k$ is larger is also of interest (for instance, from the perspective of understanding the natural dynamics in the party affiliation game).
We take the first step towards this goal by analyzing the smoothed complexity of the FLIP method for \threecut in complete graphs. 

\begin{restatable}{theorem}{threeCutCompletePoly}\label{theorem:max-3-cut-complete-poly-time}
Let $G=(V,E)$ be the complete graph on $n$ vertices, and suppose that the edge weights $(X_e)_{e\in E}$ are independent random variables chosen according to the probability density function $f_e:[-1,1]\rightarrow [0,\phi]$ for some $\phi>0$.
For every constant $\eta > 0$, with high probability every implementation of the FLIP method for \threecut terminates in at most
$O(\phi n^{99 + \eta})$ steps. 
\end{restatable}



In particular, by taking $\eta=0.1$, our theorem implies that the FLIP method for \threecut terminates in at most $O(\phi n^{99.1})$ moves\footnote{Our run-time bound could be improved, but we state a weaker bound for the purposes of simplicity in the analysis.}. 
We observe that the techniques of Angel et al. do not extend to address the smoothed complexity of the FLIP method for \threecut in complete graphs. Angel et al.'s union bound argument for the matrix that they associate with the sequence of moves crucially relies on the fact that its entries are independent of the starting $2$-cut. Unfortunately, a similar matrix for \threecut has entries that depend on the starting $3$-cut and hence, their union bound arguments fail to extend. We overcome this issue by relying on a completely different matrix. We believe that our techniques underlying Theorem \ref{theorem:max-3-cut-complete-poly-time} provide a general framework to address 
the smoothed complexity of the FLIP method for \kcut in complete graphs for any constant $k$. \\ 

Finally, we also show that the smoothed complexity of the FLIP method for \kcut in arbitrary graphs (i.e., noise is added only to the edges of the given graph) is quasi-polynomial for constant $k$. 

\begin{restatable}{theorem}{kCutQuasiPoly}\label{theorem:max-k-cut-quasi-poly}
Let $G=(V,E)$ be an arbitrary graph on $n$ vertices, and suppose that the edge weights $(X_e)_{e\in E}$ are independent random variables chosen according to the probability density function $f_e:[-1,1]\rightarrow [0,\phi]$ for some $\phi>0$. 
For every constant $\eta > 0$, with high probability every implementation of the FLIP method for \threecut terminates in at most
$O(\phi n^{2(2k-1)k \log{(kn)}+3+\eta})$ steps. 
\end{restatable}

%
%


\subsection{Related work}
The literature of smoothed analysis is vast with successful analysis of several algorithms for various problems. We mention the works closely related to \cut. Els\"{a}sser and Tscheuschner \cite{ET11} showed that if the edge weights are perturbed using Gaussian noise in graphs with maximum degree $O(\log{n})$, then the complexity of the FLIP method is polynomial. 
Etscheid and R\"{o}glin \cite{ER15} considered another special case in which the vertices are points in a $d$-dimensional space and the edge weights are given by the squared Euclidean distance between these points. In this setting, they showed that if the points are perturbed by Gaussian noise, then the complexity of the FLIP method is polynomial. After these special settings, Etscheid and R\"{o}glin \cite{ER17} showed that the smoothed complexity of the FLIP method in arbitrary graphs is quasi-polynomial. Subsequently, Angel, Bubeck, Peres and Wei \cite{ABPW17} showed that the smoothed complexity of the FLIP method in complete graphs is polynomial. 

As we mentioned in the introduction, the FLIP method is of interest as a natural dynamics towards computing a Nash equilibrium in certain games. In the \emph{non-smoothed} setting, computing a local \kcut is a special case of computing a pure Nash equilibrium in network coordination games. Concurrent to our work and independent of it, Boodaghians, Kulkarni and Mehta \cite{BKM18} have given an efficient algorithm for computing a Nash equilibrium in \emph{smoothed} network coordination games. However, it is important to note that their smoothed setting for network coordination games when specialized to the case of \kcut \emph{does not} correspond to our smoothed setting for \kcut. So, our results on the smoothed complexity of the FLIP method for local \threecut and local \kcut \emph{do not} follow from their results. Alternatively, they present a ``smoothness preserving reduction'' from computing Nash equilibrium in network coordination games involving only $2$ strategies to computing local max-cut. Our Theorem \ref{theorem:max-cut-poly} complements this result as it now follows that the smoothed complexity of a natural dynamics for computing a Nash equilibrium in $2$-strategy network coordination games is $O(n^{7.83})$.





\subsection{Preliminaries}
All graphs considered in this work are simple.
Let $H = (V,E)$ be a directed graph and let $v \in V$. 
Then we denote the outgoing neighborhood of $v$ in $H$ by $\Delta^{out}_H(v) := \{u \in V \mid vu \in E\}$ and the incoming neighborhood of $v$ in $H$ by $\Delta^{in}_H(v) := \{ u \in V \mid uv \in E\}$.
Likewise, we denote the outgoing arcs from $v$ in $H$ by $\delta_H^{out}(v) := \{ uw \in E \mid w = v\}$ and
 the incoming arcs to $v$ in $H$ by $\delta_H^{in}(v) := \{ wu \in E \mid w = v\}$. 
For a matrix $M$, we denote the element in the $i$'th row and $j$'th column of $M$ by $M[i,j]$ and we denote the $k$'th column of $M$ by $M^k$. For two vectors $a,b\in \R^n$, we denote their dot product by $\langle a,b\rangle$. 
We denote the set of integers between $1$ to an integer $n$ by $[n] := \{1, \ldots, n\}$. 
\section{Outline of our analysis}\label{sec:techniques}
We will first describe our analysis for \cut followed by the analysis for \kcut.

\subsection{Outline for \cut}
In this section, we outline our proof of Theorem \ref{theorem:max-cut-poly}.
Let $G=(V,E)$ be the complete graph on $n$ vertices and let $X\in [-1,1]^E$ be the edge weights. We will show that every sequence of improving moves of sufficiently large linear length, say $2n$, from any initial cut must increase the cut value by $\Omega(\phi^{-1}n^{-4.83})$ with high probability. As the edge weights  are bounded by at most $1$, the value of every local max-cut is at most $n^2$. Hence, the FLIP method must terminate in at most $O(\phi n^{7.83})$ moves with high probability.
We now outline our proof showing that every linear length sequence of improving moves from any initial cut must increase the cut value by $\Omega(\phi^{-1}n^{-4.83})$.

We represent a cut by a configuration in $\{\pm 1\}^V$ that represents the part of each vertex.
A move of a vertex from a given configuration can be described using a vector $\alpha\in \{0,\pm 1\}^E$ such that the increase in the cut value for that move is given by $\langle \alpha, X\rangle$. For a starting configuration $\tau_0\in \{\pm 1\}^V$, we define a matrix $P_{L,\tau_0}$ that conveniently nullifies the effect of non-moving vertices of $L$: for each pair of closest moves of a vertex $v$ in $L$, we have a column in $P_{L,\tau_0}$ which is the sum of the two vectors $\alpha_1$ and $\alpha_2$ corresponding to those two moves of $v$. The main advantage of this matrix $P_{L,\tau_0}$ is that it depends only on the starting configuration of the vertices which move in $L$ and is independent of the starting configuration of the vertices that do not move in $L$ (Proposition \ref{prop:non-moving-nullified-cut}). This feature is crucially helpful while taking a union bound later.
Our matrix $P_{L,\tau_0}$ is also implicitly used by Etscheid and R\"{o}glin.

We say that a sequence $L$ of moves is $\epsilon$-slow from an initial configuration $\tau_0$ with respect to edge-weights $X$ if each move of $L$ strictly improves the cut value and moreover, the total improvement made by $L$ to the cut-value is at most $\epsilon$. Thus, if $L$ is $\epsilon$-slow from $\tau_0$ with respect to $X$, then $\langle \alpha_t, X\rangle>0$ and $\sum_{t=1}^{\text{length}(L)}\langle \alpha_t,X \rangle \le \epsilon$, where $\alpha_t$ is the vector associated with the $t$'th step of $L$.
It follows that if $L$ is $\epsilon$-slow from $\tau_0$ with respect to $X$, then $\langle C,X\rangle >0$ for all columns $C$ in $P_{L,\tau_0}$ and moreover $\sum_{C\in \text{Columns}(P_{L,\tau_0})}\langle C, X\rangle\le 2\epsilon$ since each column
$\alpha$ participates in at most two columns
of the matrix $P_{L,\tau_0}$.
We define the following event for a sequence $L$ from a starting configuration $\tau_0$ for edge weights $X$:
\[
\mathcal{E}_{L,\tau_0,X}: \text{ $\langle C,X\rangle >0$ for all columns $C$ in $P_{L,\tau_0}$ and $\sum_{C\in \text{Columns}(P_{L,\tau_0})}\langle C, X\rangle\le 2\epsilon$}.
\]
We show that (Lemma \ref{lem:bbound}) for a fixed choice of $L$ and $\tau_0$, the probability (over the choices of $X$) that $\mathcal{E}_{L,\tau_0,X}$ happens is at most
\[
\frac{(2\phi \epsilon)^{rank(P_{L,\tau_0})}}{rank(P_{L,\tau_0})!}.
\]

Our probability bound mentioned above is stronger than the bound given by Etscheid and R\"{o}glin as well as Angel et al. The main fact that we exploit to obtain this stronger bound is that an $\epsilon$-slow sequence, by definition, improves the total sum $\sum_{C\in \text{Columns}(P_{L,\tau_0})}\langle C, X\rangle$ by at most $2\epsilon$. In contrast, previous works used $\epsilon$-slow in a weaker manner: Etscheid and R\"{o}glin only used the fact that $\langle C, X\rangle \le 2\epsilon$ for every column $C$ of the matrix $P_{L,\tau_0}$ while Angel et al. only used the fact that $\langle \alpha_t,X\rangle \le \epsilon$ for every step $t$ in $L$. We deviate from their analysis to fully exploit the definition of $\epsilon$-slowness.

Next, we need to upper bound the probability that there exists a starting configuration $\tau_0$ and a linear length sequence $L$ such that the event $\mathcal{E}_{L,\tau_0,X}$ happens. In order to attempt the natural union bound, we need a lower bound on the rank of the matrix $P_{L,\tau_0}$ for sequences of linear length. For example, if the rank is at least some constant fraction of $n$, then we may use a straightforward union bound. However, there are linear length sequences which have rank much smaller than $n$. We address this issue by focusing on \emph{critical} sequences.

A sequence $B$ is a critical sequence if $\ell(B)=1.71s(B)$ and moreover, $\ell(B')<1.71 s(B')$ for every subsequence $B'$ of $B$, where $\ell(L)$ and $s(L)$ denote the length and the number of vertices respectively in a sequence $L$. Here, the constant $1.71$ is chosen optimally to obtain the best possible run-time using our technique and the reasoning behind this choice is not insightful for the purposes of this overview.
Critical sequences were introduced by Angel et al., who also showed that every sequence of length $1.71n$ (i.e., sufficiently large linear length) contains a critical subsequence (Claim \ref{claim:critical-block-exists-cut}). Thus, it suffices to upper bound the probability that there exists a critical sequence $B$ and a starting configuration $\tau_0$ such that the event $\mathcal{E}_{B,\tau_0,X}$ happens.

Let us now fix a critical sequence $B$ and obtain an upper bound on the probability that there exists a starting configuration $\tau_0$ such that the event $\mathcal{E}_{B,\tau_0,X}$ happens.
We exploit the fact that the event $\mathcal{E}_{B,\tau_0,X}$ is independent of the starting configuration of the vertices that do not move in $B$ (as the matrix $P_{B,\tau_0}$ has this property). Thus, it suffices to perform a union bound over the starting configuration of the vertices that move in $B$. The number of such vertices is $s(B)$ and hence, the number of possible starting configurations of moving vertices is $2^{s(B)}$. Now, it remains to bound the probability that the event $\mathcal{E}_{B,\sigma_0,X}$ happens, where $\sigma_0$ is a fixed choice of the starting configuration of the vertices that move in $B$ and an arbitrarily chosen starting configuration of the vertices that \emph{do not} move in $B$.
For this, the above discussions suggests that we need a lower bound on the rank of the matrix $P_{B,\sigma_0}$ for a critical sequence $B$. We show that for a critical sequence $B$, the rank of $P_{B,\sigma_0}$ is at least $0.38s(B)$, where $s(B)$ is the number of vertices appearing in $B$ (Corollary \ref{coro:crit}). Thus, for a fixed critical sequence $B$, the probability that there exists a starting configuration $\tau_0$ such that the event $\mathcal{E}_{B,\tau_0,X}$ happens is at most
\[
2^{s(B)}\frac{(2\phi\epsilon)^{0.38s(B)}}{(0.38s(B))!}.
\]

Finally, we take a union bound over critical sequences. The number of possible critical sequences $B$ with $s(B)=s$ is at most $\binom{n}{s}s^{1.71s}$ since the length of such a critical sequence is $1.71s$. Thus, the probability that there exists a critical sequence $B$ and a starting configuration $\tau_0$ such that the event $\mathcal{E}_{B,\tau_0,X}$ happens is at most
\[
\sum_{s=1}^n \binom{n}{s}s^{1.71s}\left(2^s\frac{(2\phi\epsilon)^{0.38s}}{(0.38s)!}\right)=o(1),
\]
when $\epsilon=\phi^{-1}n^{-4.83}$ using Stirling's approximation.

Our analysis approach builds on top of two previously known ingredients: the matrix $P_{L,\tau_0}$ implicitly used by Etscheid and R\"{o}glin \cite{ER17} to nullify the effect of non-moving vertices and the notion of a critical block introduced by Angel, Bubeck, Peres and Wei \cite{ABPW17} that is helpful to show a lower bound on the rank of the relevant matrix. Our main contributions to improve the run-time analysis are (1) a tighter union bound by exploiting the full power of $\epsilon$-slowness and (2) improved rank lower bounds for critical sequences.

\subsection{Outline for \threecut and \kcut in arbitrary graphs}

In this section, we outline our proof of Theorems \ref{theorem:max-3-cut-complete-poly-time} and \ref{theorem:max-k-cut-quasi-poly}. The high-level approach is similar to the one for \cut described above. Let $G=(V,E)$ be an \emph{arbitrary} graph on $n$ vertices and let $X\in [-1,1]^E$ be the edge weights.
Our goal is to show that every sufficiently long linear length sequence of improving moves from any initial $k$-cut must increase the cut value by some non-negligible amount.

We represent a $k$-cut by a configuration in $[k]^V$ that represents the part of each vertex. Using an observation by Frieze and Jerrum \cite{FJ95}, we can again associate a vector $\alpha\in \{0,\pm 1\}^E$ with each move such that the increase in the cut value for that move is given by $\langle \alpha, X\rangle$. We emphasize that these vectors depend on the starting configuration. The analysis technique by Angel et al. for \cut does not extend to \kcut even for $k=3$ primarily due to the dependence of these vectors on the starting configuration. 


Our main tool to overcome this issue is by considering an appropriate matrix $P_{L,\tau_0}$ for \kcut. Instead of pairs of nearest moves in $L$ that was helpful for \cut, here, we define the notion of a \emph{cycle} over a vertex in $L$. A cycle over a vertex is a set of moves of that vertex which result in the vertex moving from one part and eventually returning to that same part. We note that this generalizes the concept of a pair in the \cut setting since any two nearest moves would form a cycle. Next, we define the matrix $P_{L,\tau_0}$ to have a column for every cycle in $L$ which is the sum of all the vectors which correspond to moves in that cycle. The matrix $P_{L,\tau_0}$ has the property that it is independent of the starting configuration of the vertices that do not move in $L$ (Proposition \ref{prop:non-moving-nullified}).

We use the same notion of $\epsilon$-slowness as that for \cut. Now, if $L$ is $\epsilon$-slow from an initial configuration $\tau_0$ with respect to edge weights $X$, then $\langle C,X\rangle \in (0,k\epsilon]$ for all columns $C\in P_{L,\tau_0}$ since each column
of the matrix $P_{L,\tau_0}$ is the sum of at most $k$ vectors $\alpha_{t_1}, \ldots \alpha_{t_k}$. We define the following event for a sequence $L$ from a starting configuration $\tau_0$ for edge weights $X$:
\[
\mathcal{D}_{L,\tau_0,X}: \text{ $\langle C,X\rangle \in (0,k\epsilon]$ for all columns $C$ in $P_{L,\tau_0}$}.
\]
For a fixed $L$ and $\tau_0$ (Lemma \ref{lemma:3bbound}), the probability (over the choices of $X$) that $\mathcal{D}_{L,\tau_0,X}$ happens is at most
\[
(k \phi \epsilon)^{rank(P_{L,\tau_0})}.
\]
The above result follows from Lemma A.1 in \cite{ER17}.

Next, we need a lower bound on the rank of the matrix $P_{L,\tau_0}$.
We emphasize that this needs substantially new combinatorial ideas in the form of considering the many different ways a cycle can interact with other cycles and vertices that are not part of any cycle.
\begin{itemize}
\item For the complete graph in the \threecut case, we show that the rank of a critical improving sequence $B$ is at least $(1/32)s(B)$, where $s(B)$ is the number of vertices that move in the sequence $B$ (Corollary \ref{coro:3crit}).
With this rank lower bound and the probability bound mentioned above, the rest of the analysis for the complete graph is similar to the analysis for \cut.

\item For arbitrary graphs in the \kcut case we show that the rank of an improving sequence $B$ is at least half the number of vertices that appear in some cycle of $B$ (Lemma \ref{lemma:half}).
Next, we show that a sequence $B$ of length $kn$ must have a subsequence $B'$ with at least $1/(2(2k-1)k\lg(kn))s(B)$ vertices that appear in some cycle of $B'$.
We use these results in conjunction with the probability bound mentioned above to show that every sequence of improving moves of length $kn$ from any initial $k$-cut must increase the cut value by $\Omega(\phi^{-1}n^{-2(2k-1)k \log{kn}})$.
\end{itemize}

Our main contributions are twofold: (1) we introduce the appropriate matrix $P_{L,\tau_0}$ that nullifies the effect of non-moving vertices which is crucial to perform the union bound and (2) rank lower bounds for this matrix for $k$-cut in arbitrary graphs and for $3$-cut in the complete graph based on combinatorial arguments. We believe that our techniques should also be helpful to address the smoothed complexity of FLIP for \kcut for larger constants $k$ in the complete graph.

\section{Smoothed analysis of FLIP for \cut}
\label{sec:2cut}
In this section we prove Theorem \ref{theorem:max-cut-poly}.
We begin with some notations. 

Let $G=(V,E)$ be a connected graph with $n$ vertices and let $X:E\rightarrow [-1,1]$ be an edge-weight function. 
We recall a convenient formulation of the objective function for \cut. 
We consider the space $\{\pm 1\}^V$ of configurations that define a partition of the vertex set into $2$ parts. 
For a configuration $\tau\in \{\pm 1\}^V$, we denote the part of $v$ by $\tau (v)$. 
For a configuration $\tau\in \{\pm 1\}^V$, the weight of $\tau$ is given by 
\begin{equation}
\label{eq:cutweight}
\frac{1}{2}\sum_{uv \in E} w(uv) (1 - \tau(u)\tau(v)).
\end{equation}
Let
\[
H(\tau) := -\frac{1}{2} \sum_{uv \in E} X_{uv}\tau(u)\tau(v).
\]
We observe that $H(\tau)$ is a translation of (\ref{eq:cutweight}) by the total weight of all edges and hence, it suffices to work with $H(\tau)$ henceforth. 

We analyze the run-time of the FLIP method in the smoothed framework for the complete graph. We will work with the complete graph in this section and avoid stating this explicitly henceforth. 
A \emph{flip} of a vertex $v \in V$ changes $\tau(v)$ to $-\tau(v)$. We will denote a \emph{move} by the vertex that is flipped. 
We will need the notions of improving sequences that we define now.

\begin{definition}
Let $L$ be a sequence of moves, $\tau_0\in \{\pm 1\}^V$ be an initial configuration and $X\in [-1,1]^E$ be the edge weights. We will denote the length of the sequence $L$ by $\ell(L)$, the set of vertices appearing in the moves in $L$ by $S(L)$, and $s(L):=|S(L)|$. For each $v\in V$, we will denote the number of times that the vertex $v$ moves in $L$ by $\#_L(v)$. 
We will denote the $t$'th move of $L$ by $L(t)=v_t$. 
\begin{enumerate}
\item For each $t\in [\ell(L)]$, we will denote $\tau_t$ as the configuration obtained from $\tau_{t-1}$ by setting $\tau_t(u):=\tau_{t-1}(u)$ for every $u\in V\setminus \{v_t\}$ and $\tau_t(v_t):=-\tau_{t-1}(v_t)$.

\item We say that \emph{$L$ is improving from $\tau_0$ with respect to $X$} if $H(\tau_{t})-H(\tau_{t-1})>0$ for all $t\in [\ell(L)]$. We say that \emph{$L$ is $\epsilon$-slowly improving from $\tau_0$ with respect to $X$} if 
$L$ is improving from $\tau_0$ and 
$H(\tau_{t})-H(\tau_{0})\in (0,\epsilon]$.  
\end{enumerate}
\end{definition}

Next, we obtain a convenient expression for characterizing the improvement of $H(\tau)$ in each step. 
\begin{definition}
Let $L$ be a 
sequence of moves and $\tau_0\in \{\pm 1\}^V$ be an initial configuration. Let $M_{L,\tau_0}\in \{0,\pm 1\}^{E\times [\ell(L)]}$ be a matrix with rows corresponding to the edges of $G$, columns corresponding to time-steps in the sequence $L$, and whose entries are given by 
\[
M_{L,\tau_0}[\{a,b\},t]:=
\begin{cases}
+1 &\mbox{if } L(t) \in \{a,b\} \text{ and } \tau_t(a) \not= \tau_t(b),\\
-1 &\mbox{if } L(t) \in \{a,b\} \text{ and } \tau_t(a) = \tau_t(b),\\
0 &\mbox{otherwise}, 
\end{cases}
\]
where $\{a,b\}\in E$ and $t\in [\ell(L)]$. 
\end{definition}

\begin{remark}
For a sequence $L$ from an initial configuration $\tau_0$, we have $H(\tau_t)-H(\tau_{t-1})=\langle M_{L,\tau_0}^t,X\rangle$.
\end{remark}

Next, we need the notion of repeating and singleton vertices. We note that the following definition does not depend on the initial configuration.
\begin{definition}
Let $L$ be a sequence of moves. 
We will denote the number of times that a vertex $v$ moves in $L$ by $\#_L(v)$.
A vertex $v$ is called \emph{repeating} if $\#_L(v) \geq 2$ and is called a \emph{singleton} otherwise.
Let $S_1(L)$ and $S_2(L)$ denote the set of repeating and singleton vertices of $L$ respectively, and let $s_1(L) := |S_1(L)|$ and $s_2(L) := |S_2(L)|$.
For time steps $t_1,t_2\in \ell(L)$, the ordered pair $(t_1,t_2)$ is a \emph{pair for vertex $v \in V$} if $t_1<t_2$, $L(t_1) = v = L(t_2)$ and $L(t) \not = v$ for all $t\in \{t_1+1, \ldots, t_2-1\}$.
For all $v \in V$, let $\Gamma(L, v)$ be the set of pairs for $v$ in $L$ and $\Gamma(L) := \cup_{v \in V} \Gamma(L,v)$ be the set of pairs for all vertices in $L$.
\end{definition}

We now define a suitable matrix that will nullify the influence of non-moving vertices. 
\begin{definition}
Let $L$ be a sequence of moves and let $\tau_0\in \{\pm 1\}^V$ be an initial configuration. 
Let $P_{L,\tau_0} \in  \{0,\pm 1\}^{E\times \Gamma(L)}$ be a matrix with rows corresponding to edges of $G$, columns corresponding to pairs in $L$, and whose entries are given by 
\[
P_{L,\tau_0}[\{a,b\}, C] := \sum\limits_{t \in C} M[\{a,b\}, t], 
\]
where $\{a,b\}\in E$ and $C\in \Gamma(L)$. 
\end{definition}

\begin{prop}
\label{prop:non-moving-nullified-cut}
Let $L$ be a sequence of moves. 
If $v \in V\setminus S(L)$, then $P_{L,\tau_0}[\{a,v\},C]=0$ for every $C\in \Gamma(L)$ and $\{a,v\}\in E$. 
\end{prop}
\begin{proof}
Let $C = (t_1, t_2) \in \Gamma(L)$ and $\{a,v\}\in E$. 
Since $v$ is not in $S(L)$, it follows that $C$ is not a pair for $v$. 
If $C$ is not a pair for $a$, then $M_{L,\tau_0}[\{a,v\},t_i]=0$ for $i \in [2]$ and hence $P_{L,\tau_0}[\{a,v\},C]=0$. 
Suppose $C$ is a pair for $a$. Then, it follows that $M_{L,\tau_0}[\{a,v\},t_1]=-M_{L,\tau_0}[\{a,v\},t_2]$ and hence $P_{L,\tau_0}[\{a,v\},C]=M_{L,\tau_0}[\{a,v\},t_1]+M_{L,\tau_0}[\{a,v\},t_2]=0$.
\end{proof}

\subsection{Rank lower bounds for $P_{L,\tau_0}$} 
Let $L$ be a sequence of moves and let $\tau_0$ be an initial configuration.
In this section, we show a lower bound on the rank of $P_{L, \tau_0}$. 
For this, we will make use of a directed graph with certain properties.
We define these properties now.

\begin{definition} \label{definition:good-ngbrwise-indep-cut}
Let $L$ be a sequence of moves.
\begin{enumerate}[(i)]
\item \label{item:good-cut} 
For $u,v\in S(L)$, we will call the ordered pair $vu$ to be an \emph{$L$-good-arc} if there exists a pair $C \in \Gamma(L)$ for $v$ such that $P_{L,\tau_0}[\{u,v\},C]\neq 0$. 
A directed graph $H$ whose 
nodes
are a subset of $S(L)$, is \emph{$L$-good} if every arc in $H$ is an $L$-good-arc. 
\item \label{item:ngbrwise-indep-cut} A directed graph $H$ with 
node
set $S(L)$ is \emph{$L$-neighbor-wise independent} if for every $v\in S(L)$, there exists an ordering of $\Delta_H^{out}(v)$, say $u_1,\ldots, u_m$, along with pairs $C_1,\ldots, C_m\in \Gamma(L)$ for $v$ such that
\begin{enumerate}
\item $P_{L,\tau_0}[\{v,u_i\},C_i]\neq 0$ and
\item $P_{L,\tau_0}[\{v,u_j\},C_i]= 0$ for every $j\in \{i+1,\ldots, m\}$. 
\end{enumerate}
\item A directed graph $H$ with 
node
set $S(L)$ is \emph{functional} if $|\delta^{out}_H(v)|=1$ for every $v\in S_2(L)$.
\end{enumerate}
\end{definition}

\begin{remark}
For $u,v\in S(L)$, the ordered pair $vu$ is an $L$-good-arc if and only if there exists a pair $(t_1, t_2) \in \Gamma(L)$ for $v$ such that the number of times $u$ appears between $t_1$ and $t_2$ is odd.
\end{remark}

The following lemma is our key tool in obtaining a lower bound on the rank of $P_{L, \tau_0}$.
We show that the rank is at least the number of arcs in an $L$-good $L$-neighbor-wise independent directed acyclic graph.

\begin{lemma} \label{lemma:hrank-cut}
Let $L$ be a sequence of moves and let $\tau_0 \in \{\pm 1\}^V$ be an initial configuration.
Let $H$ be an $L$-good $L$-neighbor-wise independent directed acyclic graph.
Then, 
\[
rank(P_{L, \tau_0}) \geq |E(H)|.
\]
\end{lemma}

\begin{proof}
Consider the submatrix $B_H$ of $P_{L, \tau_0}$ consisting of the rows corresponding to edges $\{u,v\}$ for every arc $vu \in E(H)$. 
We now show that the matrix $B_H$ has full row-rank by induction on $|E(H)|$. 
The base case for $|E(H)|=0$ is trivial.

For the induction step, consider $|E(H)|\ge 1$. Suppose that there exist coefficients $\mu_{\{u,v\}}\in \R$ for every $(u,v)\in E(H)$ such that
\[
\sum_{uv\in E(H)}\mu_{\{u,v\}}P_{L, \tau_0} [\{u,v\},(t,t')] = 0
\]
for every pair $(t,t') \in \Gamma(L)$.
Since $H$ is a directed acyclic graph with at least one arc, there exists a
node
$v\in V(H)$ with $|\delta^{out}_H(v)| \geq 1$ and $|\delta^{in}_H(v)| = 0$.

\begin{claim}
For every $u\in \Delta^{out}_H(v)$, the coefficient $\mu_{\{v,u\}}$ is zero.
\end{claim}
\begin{proof}
Consider the ordering $u_1,\ldots, u_m$ of the 
vertices
in $\Delta^{out}_H(v)$ and pairs $C_i \in \Gamma(L)$ satisfying Definition \ref{definition:good-ngbrwise-indep-cut} (\ref{item:ngbrwise-indep-cut}). We show that $\mu_{\{v,u_j\}}=0$ for every $j\in[m]$ by induction on $j$. 

For the base case, we consider $j=1$. Consider the column of $P$ corresponding to the pair $C_1$. Since $C_1$ is a pair for $v$ and the vertex $v$ has no incoming arcs in $H$, the only possible non-zero entries in this column among the chosen rows are in the rows corresponding to the edges $\{v,u_1\},\ldots, \{v,u_m\}$. Thus,
\begin{align*}
0
&=\sum_{vu\in E(H)}\mu_{\{v,u\}}P_{L,\tau_0}[{\{v,u\},C_1}]
=\sum_{i=1}^m\mu_{\{v,u_i\}}P_{L,\tau_0}[{\{v,u_i\},C_1}]. \label{eq:base-case-1-cut}
\end{align*}
Moreover, by the choice of $C_1$, we have that 
\begin{align*}
P_{L,\tau_0}[\{v,u_1\},C_1] &\neq 0 \text{, and}\\
P_{L,\tau_0}[\{v,u_i\}, C_1] &= 0 \ \forall i \in [m] \setminus \{1\}.
\end{align*}
Consequently, we obtain that $\mu_{\{v,u_1\}}=0$. 
For the induction step, consider $j\ge 2$ and the columns of $P_{L, \tau_0}$ corresponding to the pair $C_j \in \Gamma(L)$. Since $C_j$ is also a pair for $v$, the only possible non-zero entries in this column among the chosen rows are in the rows corresponding to the edges $\{v,u_1\},\ldots, \{v,u_m\}$. 
Thus,
\begin{align*}
0
&=\sum_{vu\in E(H)}\mu_{\{v,u\}}P_{L,\tau_0}[{\{v,u\},C_j}]
=\sum_{i=1}^m\mu_{\{v,u_i\}}P_{L,\tau_0}[{\{v,u_i\},C_j}].
\end{align*}
By the induction hypothesis, we know that $\mu_{\{v,u_i\}}=0$ for every $i\in\{1,2,\ldots, j-1\}$.
Thus, 
\begin{align*}
0 &= \sum_{i=j}^m\mu_{\{v,u_i\}}P_{L,\tau_0}[{\{v,u_i\},C_j}].
\end{align*}
Moreover, by the choice of $C_j$, we also have that 
\begin{align*}
P_{L,\tau_0}[\{v,u_j\},C_j] &\neq 0 \text{, and}\\
P_{L,\tau_0}[\{v,u_i\}, C_j] &= 0 \ \forall i \in \{j+1,\ldots,m\}.
\end{align*}
Consequently, we obtain that $\mu_{\{v,u_j\}}=0$. 
\end{proof}

As a consequence of the claim, we have that the matrix $B_H$ has full row-rank if and only if the matrix $B_{H'}$ obtained from the graph $H':=H-\delta_H^{out}(v)$ has full row-rank. 
We note that the graph $H'$ is also an $L$-good $L$-neighbor-wise independent directed acyclic graph. Moreover, $|E(H')|<|E(H)|$. Thus, by the induction hypothesis, the matrix $B_{H'}$ has full row-rank.
Hence, the matrix $B_H$ also has full row-rank. 
\end{proof}

In order to show a lower bound on the rank of $P_{L, \tau_0}$ we switch our attention from $L$ to specific \emph{blocks} of $L$ with special properties.
We will see later that all sufficiently long $L$ must contain such blocks.

\begin{definition}
Let $L$ be a sequence of moves. 
A \emph{block} of $L$ is a continuous subsequence of $L$. For $t_1,t_2\in [\ell(L)]$, we will denote the block of $L$ between time-steps $t_1$ and $t_2$ as $L[t_1,t_2]$. For a block $B$ of $L$, we will denote the starting and ending time-step of $B$ in $L$ as $t_{beg}(B)$ and $t_{end}(B)$. 
Let $\beta > 0$ be a constant. A block $B$ in $L$ is $\beta$-{\it critical} if $\ell(B) \geq (1+\beta)s(B)$ and $\ell(B') < (1+\beta)s(B')$ for every block $B'$ that is strictly contained in $B$.
\end{definition}

\begin{claim} \label{clm:exist_singles}
Let $B$ be a $\beta$-critical block for $0 < \beta \leq 1$. Then, for every block $B'$ that is strictly contained in $B$, we have that $S_1(B')\neq \emptyset$. 
\end{claim}

\begin{proof}
Let $B'$ be a block that is strictly contained in $B$. 
Suppose for the sake of contradiction that every vertex in $B'$ appears at least twice. 
Then, $\ell(B')\ge 2s(B') \ge (1+\beta)s(B')$, which contradicts the criticality of $B$.
\end{proof}

We next show that for a $\beta$-critical block $B$, there exists a path of $B$-good-arcs from any repeating vertex 
to a singleton vertex. 

\begin{lemma} \label{lem:s2}
Let $B$ be a $\beta$-critical block for $0 < \beta \leq 1$ with $S_1(B)\neq \emptyset$. 
Then, for every $v \in S_2(B)$, there exists a sequence $u_0u_1, u_1 u_2, \ldots, u_{k} u_{k+1}$ of $L$-good-arcs such that $u_0 = v$ and $u_{k+1} \in S_1(B)$.
\end{lemma}

\begin{proof}

Let $v\in S_2(B)$. We will build a chain of blocks $B_0\subsetneq B_1\subsetneq B_2\subsetneq\ldots\subsetneq B_k\subseteq B$ and vertices $u_0, u_1,\ldots, u_k,u_{k+1}$ such that 
\begin{enumerate}[(i)]
\item $u_0=v$ and the vertices $u_0,u_1,\ldots,u_k$ are distinct,
\item $u_{k+1}\in S_1(B)$,
\item $u_i\in S_1(B_{i-1})$ for every $i\in [k+1]$,
\item the end vertices of the sequence $B_i$ are $u_i$ and a vertex in $\{u_0,u_1,\ldots, u_{i-1}\}$, and 
\item for every $i\in [k+1]$, there exists an $L$-good-arc $u_ju_i$ for some $j<i$. 
\end{enumerate}

We observe that such a chain immediately implies the existence of the required sequence of $L$-good-arcs that proves the lemma:
 by condition (v), there exists a sequence of $L$-good-arcs that form a path which ends at $u_{k+1}$ and necessarily starts at $u_0$. By condition (i), we have that $u_0=v$ and by condition (ii), we have that $u_{k+1}\in S_1(B)$. 

\begin{figure}[ht]
\noindent
\rule{\textwidth}{1pt}
\textbf{Procedure}
\vspace{1mm}
\hrule
\vspace{1mm}
{\bf Input: } $\beta$-critical block $B$  ($0 < \beta \leq 1$) with $S_1(B)\neq \emptyset$ and a vertex $v\in S_2(B)$.
\begin{enumerate}
\item Initialize $i\leftarrow 0$, $u_0\leftarrow v$, 
$B_0\leftarrow B[t_1,t_2]$ for an arbitrary pair $(t_1,t_2)$ for $u_0$. 
\item $u\leftarrow$ an arbitrary vertex from $S_1(B_0)$ and $r_u\leftarrow$ time-step $r\in \{t_{beg}(B_0),\ldots, t_{end}(B_0)\}$ for which $B(r)=u$. 
\item While $(u\not\in S_1(B))$:
\begin{enumerate}
\item $i\leftarrow i+1$.
\item $u_{i}\leftarrow u$.
\item 

Let $(r_u,q_u)$ be a pair for $u$. Since $u\not\in S_1(B)$, but $u\in S_1(B_{i-1})$ it follows that $q_u\not\in [t_{beg}(B_{i-1}),t_{end}(B_{i-1})]$. Set
\[
B_{i-1}'\leftarrow
\begin{cases}
B[q_u,t_{beg}(B_{i-1})-1] &\text{ if } q_u<t_{beg}(B_{i-1}), \\
B[t_{end}(B_{i-1})+1,q_u] &\text{ if } q_u>t_{end}(B_{i-1}). 
\end{cases}
\]

\item $B_{i}\leftarrow$ concatenation of $B_{i-1}$ and $B_{i-1}'$ as they appear in $B$.
\item $u\leftarrow$ an arbitrary vertex from $S_1(B_{i})$ and $r_u\leftarrow$ time-step $r\in \{t_{beg}(B_i),\ldots, t_{end}(B_i)\}$ for which $B(r)=u$.
\end{enumerate}
\item $u_{i+1}\leftarrow u$.
\end{enumerate}
\rule{\textwidth}{1pt}
\caption{Procedure for Lemma \ref{lem:s2}.}
\label{fig:procedure}
\end{figure}

We now show that a chain of blocks and vertices satisfying properties (i)-(v) indeed exist. For this, we consider the procedure given in Figure \ref{fig:procedure}. 
We first note that a valid choice for the vertex $u$ in steps 2 and 3(e) exist by Claim \ref{clm:exist_singles} and by the fact that $S_1(B)\neq \emptyset$. 
The sequence constructed by the procedure is a nested chain since $B_{i+1}$ is a concatenation of $B_i$ with a block adjacent to $B_i$. Moreover, by steps 3(c) and 3(d), the block $B_{i-1}$ is strictly contained in the block $B_{i}$. Thus, the procedure will indeed terminate since the sequence of blocks $B_0,B_1,\ldots, B_k$ forms a nested chain that grows in size until $B_k=B$ at which point step 3(e) finds $u\in S_1(B_k)=S_1(B)$. We now show that the required conditions are satisfied. Let the procedure terminate with the nested chain $B_0\subsetneq B_1\subsetneq B_2\subsetneq \ldots \subsetneq B_k\subseteq B$ and vertices $u_0,u_1,\ldots, u_k, u_{k+1}$.

\begin{enumerate}[(i)]
\item By initialization $u_0=v$. Since $B_i$ is strictly contained in $B_{i+1}$ and by steps 3(d) and 3(a), the vertices $u_0,u_1,\ldots,u_k$ are distinct.
\item By the termination criteria, we have that $u_{k+1}\in S_1(B)$.
\item By steps 3(e), 3(a) and 3(b), we have that $u_i\in S_1(B_{i-1})$ for every $i\in [k+1]$.
\item By steps 3(c) and 3(d), the end vertices of the sequence $B_i$ are the vertex $u_i$ and a vertex in $\{u_0,u_1,\ldots, u_{i-1}\}$.
\item Let $i\in [k+1]$ and consider $u_i$. Let $h$ be the smallest index in $\{0,1,\ldots, k+1\}$ such that $u_i\in S_1(B_h)$. Then, by condition (iii), we have that $h\le i-1$. We will show that the arc $u_hu_i$ is an $L$-good-arc. 
If $h=0$, then, by the choice of $B_0$, we have that the arc $u_0u_i$ is an $L$-good-arc. So, we may assume that $h\ge 1$. Then, the block $B_h$ is a concatenation of $B_{h-1}$ and $B_{h-1}'$. Also, by the choice of $h$, we know that $u_i\not\in B_{h-1}$. Hence, the vertex $u_i\in B_{h-1}'$. Now, by condition (iv), we have that the block $B_{h-1}'$ ends with $u_h$. By condition (iii), we have that $u_h\in S_1(B_{h-1})$. Hence, we have an $L$-good-arc $u_hu_i$. 
\end{enumerate}
\end{proof}

\begin{definition}
Let $L$ be an improving sequence from some initial configuration.
Then a maximal block of repeating vertices is called a \emph{transition block} and a maximal block of singletons is called a \emph{singleton block}.
\end{definition}

For a sequence $L$ and a vertex $v\in S_2(L)$, let $b_L(v)$  denote the number of transition blocks that contain $v$. 
Also, let $R(L):=\{v\in S_2(L): b_L(v)\ge 2\}$ (i.e., $R(L)$ is the set of vertices which appear in at least two transition blocks) and let $r(L):=|R(L)|$. 

\begin{lemma} \label{lemma:block_lower_bound}

Let $B$ be a $\beta$-critical block where $0<\beta\leq 1$ with $S_1(B)\neq \emptyset$. Let $\tau_0\in \{\pm 1\}^V$ be an initial configuration. 
Then, 
\[
rank(P_{B,\tau_0}) \geq s_2(B) - r(B) + \sum_{v \in S_2(B)}(b_B(v)-1).
\]
\end{lemma}

\begin{proof}

We first build a functional $B$-good directed acyclic graph with $s_2(B)$ arcs as follows: consider the $B$-good directed graph $H$ over the node set $S(B)$ containing all possible $B$-good-arcs. Now, run a reverse breadth first search from the nodes in $S_1(B)$ and consider the subgraph $H'$ over the node set $S(B)$ obtained by including only the reverse-BFS-tree arcs. By Lemma \ref{lem:s2}, the reverse BFS search traverses all nodes in $S_2(B)$ and hence, the graph $H'$ has at least $s_2(B)$ arcs. Moreover, the graph $H'$ is a functional graph and acyclic since we included only the reverse-BFS-tree arcs. Now, we note that every functional $B$-good directed graph is $B$-neighbor-wise independent as every node has only one outgoing arc. Thus, the graph $H'$ is a $B$-good $B$-neighbor-wise independent directed acyclic graph.

Next, we add more arcs to $H'$ to obtain a larger $B$-good $B$-neighbor-wise independent directed acyclic graph $H''$ as follows: Consider a node $v\in R(B)$. We have exactly one outgoing arc $vw$ from $v$ in $H'$. For every two adjacent transition blocks containing $v$, we may add one $B$-good-arc $vu$ to a vertex $u$ that appears in a singleton block between those two adjacent transition blocks. At most one newly added arc for $v$ is not $B$-neighbor-wise independent with $vw$. We discard this arc. We perform the above operation for every node $v \in R(B)$. The resulting graph $H''$ is a $B$-good $B$-neighbor-wise independent directed acyclic graph. The number of arcs in $H''$ is at least 
\[
s_2(B)+\sum_{v\in R(B)}(b_B(v)-2) = s_2(B)-r(B)+\sum_{v\in S_2(B)}(b_B(v)-1).
\]
Then, by Lemma \ref{lemma:hrank-cut}, we have that 
\[
rank(P_{B, \tau_0}) \geq s_2(B) - r(B) + \sum_{v \in S_2(B)}(b_B(v)-1).
\]
\end{proof}

\begin{corollary} \label{coro:crit}

Let $B$ be a $\beta$-critical block where $0<\beta\leq 1$ with $S_1(B)\neq \emptyset$. Let $\tau_0\in \{\pm 1\}^V$ be an initial configuration. 
Then, 
\[
rank(P_{B,\tau_0}) \geq \frac{\beta}{1+2\beta}s(B).
\]
\end{corollary}

\begin{proof}

We first show that $rank(P_{B,\tau_0})\ge \frac{\beta}{1+\beta}s_1(B)$. 
Let $T_1,\ldots, T_k$ denote the transition blocks of $B$,
and for a vertex $v$, let $1_{v\in T_i}$ denote the indicator function for whether the vertex $v$ appears in $T_i$. Then,
\[
\sum_{v\in S_2(B)}b_B(v) = \sum_{v\in S_2(B)}\sum_{i=1}^k 1_{v\in T_i}=\sum_{i=1}^k \sum_{v\in S_2(B)}1_{v\in T_i} = \sum_{i=1}^k s(T_i).
\]
Since each $T_i$ is a proper sub-block of $B$, we have that $\ell(T_i)<(1+\beta)s(B_i)$ and hence,
\[
\sum_{v\in S_2(B)}b_B(v) = \sum_{i=1}^k s(T_i) \ge \frac{1}{1+\beta}\sum_{i=1}^k \ell(T_i).
\]

By Lemma \ref{lemma:block_lower_bound}, we have that 
\begin{align*}
rank(P_{B,\tau_0})
&\ge s_2(B)-r(B) + \sum_{v\in S_2(B)}(b_B(v)-1)\\
&\ge \sum_{v\in S_2(B)}(b_B(v)-1) \quad \quad \quad \text{(since $r(B)\le s_2(B)$)}\\
&= \sum_{v\in S_2(B)} b_B(v) - s_2(B)\\
&\geq \frac{1}{1+\beta}\sum_{i=1}^k \ell(T_i) - s_2(B) 
\\
&= \frac{\ell(B)-s_1(B)}{1+\beta} - s_2(B)\\
&= \frac{\ell(B)}{1+\beta}-\frac{1}{1+\beta}s_1(B) - s_2(B)\\
&\ge \frac{\beta}{1+\beta}s_1(B). \quad \quad \quad \text{(since $B$ is a critical block, we have $\ell(B) \geq (1+\beta)s(B)$)}
\end{align*}
Moreover, by Lemma \ref{lemma:block_lower_bound},
\begin{align*}
rank(P_{B,\tau_0})
&\ge s_2(B)-r(B) + \sum_{v\in S_2(B)}(b_B(v)-1)\\
&= s_2(B) + \sum_{v\in R(B)}(b_B(v)-2)\\
&\ge s_2(B). \quad \quad \quad \text{(since $b_B(v)\ge 2$ for every $v\in R(B)$)}
\end{align*}
Thus,
\[
rank(P_{B, \tau_0})\ge \max\left\{s_2(B),\frac{\beta}{1+\beta}s_1(B)\right\}.
\]
Let $\lambda:=s_1(B)/s(B)$. Then, $s_1(B)=\lambda s(B)$ and $s_2(B)=(1-\lambda)s(B)$ with $\lambda\in (0,1]$. 
Therefore, 
\[
rank(P_{B, \tau_0})
\ge \max\left\{1-\lambda,\frac{\beta}{1+\beta}\lambda \right\} s(B)
\ge \frac{\beta}{1+2\beta}s(B).
\]
\end{proof}

\subsection{Run-time of FLIP for \cut in the complete graph}
In this subsection we will show that a linear-sized improving sequence will improve the value of $H(\tau)$ by some non-negligible amount with constant probability. 
Theorem \ref{theorem:max-cut-poly} will follow from this result.
First we show a slight extension to a lemma of \cite{ER17}.
The proof is presented in the appendix (See Appendix \ref{app:prob}).

\begin{restatable}{lemma}{lemmakbbound}
\label{lem:prob}
Let $\phi > 0$ and $X_1,\dots,X_m$ be independent random variables with
density functions $f_1,\dots,f_m: \R \rightarrow [0,\phi]$.
 Let
$X:=(x_1,\dots,x_m)^\intercal$
and $\alpha_1,\dots,\alpha_k \in \Z^n$ be linearly independent vectors. Then for every $\epsilon > 0$,
\[
\Pr_X \Big[\langle \alpha_i, X \rangle > 0 \ \forall i \in [k] \ \text{and} \ \sum_{i=1}^k \langle \alpha_i, X \rangle \leq \epsilon \Big] \leq \frac{(\phi \epsilon)^k}{k!}.
\]
\end{restatable}

We will use the above lemma and the rank lower bound to bound the probability of the existence of a bad starting configuration for critical blocks. 

\begin{lemma}
\label{lem:bbound}
Let $B$ be a $\beta$-critical block where $0 < \beta \leq 1$ with $S_1(B)\neq \emptyset$. Then
\[
\Pr_X[\exists\ \tau_0 \in \{\pm 1\}^V : B \text{ is $\epsilon$-slowly improving from $\tau_0$ with respect to $X$}] \leq 2^{s(B)} \frac{(2\phi \epsilon)^{\frac{\beta}{1 + 2\beta}s(B)}}{\left(\frac{\beta}{1 + 2\beta}s(B)\right)!}.
\]
\end{lemma}

\begin{proof}
Suppose that $B$ is $\epsilon$-slowly improving from some $\tau_0$ with respect to $X$.
Since $B$ is improving, we have that $\langle M_{B,\tau_0}^t, X \rangle > 0$ for all $t \in [\ell(B)]$.
Since $B$ is $\epsilon$-slow, we have that 
$\sum_{t = 1}^{\ell(B)} \langle M_{B,\tau_0}^t X \rangle \leq \epsilon$.
As every column of $P_{B,\tau_0}$ is the sum of two columns of $M_{B,\tau_0}$,
we have $\langle P_{B,\tau_0}^C, X \rangle > 0$ for all pairs $C \in \Gamma(L)$.
Moreover, 
every column of $M_{B, \tau_0}$ contributes at most to two distinct columns of $P_{B, \tau_0}$.
Hence $\sum_{C \in \Gamma(L)} \langle P_{B,\tau_0}^t, X\rangle \leq 2\epsilon$.

%

For $\pi_{f}:S(B)\rightarrow \{\pm 1\}$ and $\pi_{c}:V\setminus S(B)\rightarrow \{\pm 1\}$, let us define $\tau_{(\pi_f,\pi_c)}:V\rightarrow \{\pm 1\}$ as 
\[
\tau_{(\pi_f,\pi_c)}(u):=
\begin{cases}
\pi_f(u)&\text{ if $u\in S(B)$ and}\\
\pi_c(u)&\text{ if $u\in V\setminus S(B)$}.
\end{cases}
\]
Let $\mathcal{R}_{B, \tau, X}$ denote the event that $B$ is $\epsilon$-slowly improving from the initial configuration $\tau$ with respect to $X$. 
Then, by union bound, the required probability is at most 
\[
\sum_{v\in S(B)}\sum_{\pi_f(v)\in \{\pm 1\}} \Pr_X \big[\exists\ \pi_c:V\setminus S(B)\rightarrow \{\pm 1\}:\ \mathcal{R}_{B, \tau_{(\pi_f,\pi_c)}, X}  \big].
\]

Now, consider a fixed choice of $\pi_f:S(B)\rightarrow \{\pm 1\}$. We would like to bound the following probability: 
\[
\Pr_X \big[\exists\ \pi_c:V\setminus S(B)\rightarrow \{\pm 1\}:\ \mathcal{R}_{B, \tau_{(\pi_f,\pi_c)}, X} \big].
\]
Let us define an initial configuration $\overline{\pi}_c:V\setminus S(B)\rightarrow \{\pm 1\}$ by $\overline{\pi}_c(u)=1$ for all $u\in V\setminus S(B)$ and consider $\sigma:=\tau_{(\pi_f,\overline{\pi}_c)}$. 
By Proposition \ref{prop:non-moving-nullified-cut}, we have $P_{L,\sigma} = P_{L,\tau_{(\pi_f,\pi_c)}}$ for every $\pi_c:V\setminus S(B)\rightarrow \{\pm 1\}$.
Hence, 
\[
\Pr_X \big[\exists\ \pi_c:V\setminus S(B)\rightarrow \{\pm 1\}:\ \mathcal{R}_{B, \tau_{(\pi_f,\pi_c)}, X} \big]
= \Pr_X\Big[ \mathcal{R}_{B, \sigma, X} \Big] 
\leq \frac{(2\phi \epsilon)^{\frac{\beta}{1 + 2\beta}s(B)}}{\left(\frac{\beta}{1 + 2\beta}s(B)\right)!}.
\]
The last inequality above follows from Lemma \ref{lem:prob} and Corollary \ref{coro:crit}. 
Hence, the required probability is at most
\[
2^{s(B)}\frac{(2\phi \epsilon)^{\frac{\beta}{1 + 2\beta}s(B)}}{\left(\frac{\beta}{1 + 2\beta}s(B)\right)!}.
\]
\end{proof}

\begin{lemma}\label{lemma:epsimp-2}
Let $G$ be the complete graph with $|V| = n$, let $\epsilon:=e^{-\frac{2(1+2\beta)}{\beta}}\phi^{-1}n^{-\left(\frac{1+2\beta+2\beta^2}{\beta}+\frac{\eta(1+2\beta)}{\beta}\right)}$ for a constant $\eta>0$, and let $\beta\in (0,1)$. 
Then, the probability (over the choices of $X$) that there exists a sequence $L$ of moves of length $\lceil (1+\beta)n \rceil$ and an initial configuration $\tau_0\in \{\pm 1\}^V$ such that $L$ is $\epsilon$-slowly improving from $\tau_0$ with respect to $X$ is $o(1)$. 
\end{lemma}

\begin{proof}
Let $\mathcal{R}_X$ denote the event that there exists a sequence $L$ of moves length $\lceil (1+\beta)n \rceil$ and an initial configuration $\tau_0\in \{\pm 1\}^V$ and such that $L$ is $\epsilon$-slowly improving from $\tau_0$ with respect to $X$. The following two claims show that if $\mathcal{R}_X$ happens, then there exists a starting configuration $\tau_0\in \{\pm 1\}^V$ and a $\beta$-critical block $B$ such that $B$ is $\epsilon$-slowly improving from $\tau_0$ with respect to $X$ and moreover, $\ell(B)=\lceil(1+\beta)s(B)\rceil$ and $S_1(B)\neq \emptyset$. 

\begin{claim}\label{claim:critical-block-exists-cut}
Let $L$ be a sequence of moves of length $\lceil (1+\beta)n\rceil$ for some $\beta>0$. Then there exists a $\beta$-critical block $B$ in $L$ such that $\ell(B)=\lceil (1+\beta)s(B)\rceil$. 
\end{claim}
\begin{proof}
Consider an inclusion-wise minimal block $B$ in $L$ such that $\ell(B)\ge (1+\beta)s(B)$. 
We note that such a block exists since $\ell(L)=\lceil(1+\beta)n\rceil\ge (1+\beta)s(B)$. By inclusion-wise minimality, we have that $B$ is a $\beta$-critical block. 
Suppose $\ell(B)\ge \lceil (1+\beta)s(B)\rceil + 1$. 
Then, consider the sub-block $B'\subsetneq B$ obtained by removing the last vertex from $B$. For this sub-block, we have that $s(B')\le s(B)$ and $\ell(B')=\ell(B)-1\ge \lceil (1+\beta)s(B)\rceil\ge \lceil(1+\beta)s(B')\rceil\ge (1+\beta)s(B')$, thus, contradicting the choice of $B$. 
\end{proof}

\begin{claim}\label{claim:critical-block-has-non-empty-singleton-cut}
Let $B$ be a $\beta$-critical block where $0<\beta<1$ with $\ell(B)=\lceil (1+\beta)s(B)\rceil$. Suppose there exists an initial configuration $\tau_0\in \{\pm 1\}^V$ and edge weights $X\in [-1,1]^E$ such that $B$ is improving from $\tau_0$ with respect to $X$. 
Then, $S_1(B)\neq \emptyset$. 
\end{claim}
\begin{proof}
Suppose for the sake of contradiction that $S_1(B)=\emptyset$. Suppose $2s(B)>\lceil (1+\beta)s(B)\rceil$.
Since every vertex appears at least twice in the sequence $B$, we have
\[
\ell(B)\ge 2s(B)>\lceil (1+\beta)s(B)\rceil=\ell(B), 
\] 
a contradiction. 
Therefore, we may assume that $2s(B) = \lceil (1+\beta)s(B)\rceil$ (i.e., $\ell(B)=2s(B)$). 
Let $S_{even}(B)$ and $S_{odd}(B)$ denote the set of vertices which appear even and odd number of times in the sequence $B$ respectively. 
Since $B$ is improving from $\tau_0$ with respect to $X$, it follows that $S_{odd}(B)\neq \emptyset$ (otherwise, every vertex moves even number of times in $B$ which means that the final configuration is the same as the initial configuration $\tau_0$ and consequently, the sequence $B$ would not have been improving). Now, we have 
\begin{align*}
\ell(B) 
&= \sum_{v\in S_{even}(B)}\#_B(v) + \sum_{v\in S_{odd}(B)}\#_B(v) \\
&\ge 2|S_{even}(B)| + 3|S_{odd}(B)| \quad \quad \text{(since $S_1(B)=\emptyset$ by assumption)}\\ 
&= 2s(B) + |S_{odd}(B)| \quad \quad \text{(since $|S(B)|=|S_{even}(B)|+|S_{odd}(B)|$)}\\
&> \ell(B) \quad \quad \text{(since $\ell(B)=2s(B)$ and $|S_{odd}(B)|\ge 1$),}
\end{align*}
a contradiction. 

\end{proof}


Let $\mathcal{B}$ be the set of $\beta$-critical blocks with $\ell(B)\le \lceil(1+\beta)n\rceil$ and $S_1(B)\neq \emptyset$. Then, 
\begin{align*}
\Pr_X[\mathcal{R}_X]
&\le \sum_{B\in \mathcal{B}} \Pr_X[\exists\ \tau_0 \in \{\pm 1\}^V : B \text{ is $\epsilon$-slowly improving from $\tau_0$ with respect to $X$}] \\
&\le \sum_{B\in \mathcal{B}} 2^{s(B)} \frac{(2\phi \epsilon)^{\frac{\beta}{1 + 2\beta}s(B)}}{\left(\frac{\beta}{1 + 2\beta}s(B)\right)!} \quad \quad \text{(By Lemma \ref{lem:bbound})}\\
&\le \sum_{s=1}^n \binom{n}{s}s^{\lceil(1+\beta)s\rceil} 2^{s} \frac{(2\phi \epsilon)^{\frac{\beta}{1 + 2\beta}s}}{\left(\frac{\beta}{1 + 2\beta}s\right)!}.
\end{align*}
The last inequality above is because $\ell(B)=\lceil(1+\beta)s(B)\rceil$ and hence the number of possibilities for $B\in \mathcal{B}$ with $s(B)=s$ for a fixed $s$ is at most $\binom{n}{s}s^{\lceil(1+\beta)s\rceil}$. Now, by using Stirling's approximation and the fact that $\lceil(1+\beta)s\rceil\le (1+\beta)s+1$, we have that
\begin{align*}
\Pr_X[\mathcal{R}_X]
&\le \sum_{s=1}^n \left(\frac{ne}{s}\right)^s s^{(1+\beta)s+1}2^{s} \frac{(2\phi \epsilon)^{\frac{\beta}{1 + 2\beta}s}}{\left(\frac{\beta}{e(1 + 2\beta)}s\right)^{\frac{\beta}{1+2\beta}s}}\\
&\le \sum_{s=1}^n s\left(2^{\frac{1+3\beta}{1+2\beta}}e^{\frac{1+\beta}{1+2\beta}} \left(\frac{\beta}{1+2\beta}\right)^{-\frac{\beta}{1+2\beta}} n s^{\frac{2\beta^2}{1+2\beta}} (\phi\epsilon)^{\frac{\beta}{1 + 2\beta}}\right)^s\\
&\le \sum_{s=1}^n s\left(e^{2} \left(\frac{\beta}{1+2\beta}\right)^{-\frac{\beta}{1+2\beta}} n^{\frac{1+2\beta+2\beta^2}{1+2\beta}} (\phi\epsilon)^{\frac{\beta}{1 + 2\beta}}\right)^s. 
\end{align*}
The last inequality above is by using the fact that $s\le n$. Now, for the choice of 
\[
\epsilon=e^{-\frac{2(1+2\beta)}{\beta}}\phi^{-1}n^{-\left(\frac{1+2\beta+2\beta^2}{\beta}+\frac{\eta(1+2\beta)}{\beta}\right)},
\]
the above sum is an arithmetic-geometric sum. That is, 
\[
\Pr_X[\mathcal{R}_X]\le \sum_{s=1}^n sn^{-\eta s}\le \sum_{s=1}^{\infty} sn^{-\eta s} = \frac{n^{-\eta}}{(1-n^{-\eta})^2}
\]
which tends to $0$ as $n\rightarrow \infty$.

\end{proof}

%


We now restate and prove Theorem \ref{theorem:max-cut-poly}. 
\thmMaxCutPoly*

\begin{proof}

We will use Lemma \ref{lemma:epsimp-2} with an optimal setting of $\beta$. We will derive this optimal setting in the end. For now, let us consider $\beta\in (0,1)$.

Let $\mathcal{R}_X$ denote that event that an implementation of FLIP starting from some initial configuration $\tau_0$ follows a sequence $L$ of length
\begin{equation}
\ell(L)  \ge e^{\frac{2(1+2\beta)}{\beta}}(1+\beta) \phi n^{\left(3+\frac{1+2\beta+2\beta^2}{\beta}+\frac{\eta(1+2\beta)}{\beta}\right)}.
\label{eq:length-expression-cut}
\end{equation}

Suppose $\mathcal{R}_X$ happens. For $1 \leq i \leq z := {\ell(L)}/{(1+\beta)n}$, let $L_i$ denote the block of $L$ from time-step $(i-1)(1+\beta)n + 1$ to time-step $i(1+\beta)n $ and let $\tau_i$ denote the configuration before time-step $(i-1)(1+\beta)n + 1$. We note that 
\begin{equation}
z\ge e^{\frac{2(1+2\beta)}{\beta}}\phi n^{\left(2+\frac{1+2\beta+2\beta^2}{\beta}+\frac{\eta(1+2\beta)}{\beta}\right)}.
\label{eq:z-lower-bound-cut}
\end{equation}
For every $i\in [z]$, we have that $\ell(L_i)= \lceil (1+\beta)n \rceil$ and $L_i$ is an improving sequence from the initial configuration $\tau_i$ with respect to $X$. We will now show that there exists $i\in [z]$ such that $L_i$ is an $\epsilon$-slowly improving sequence from the initial configuration $\tau_i$ with respect to $X$ for an appropriate choice of $\epsilon$. 

For notational convenience, let $h(L)$ denote the total improvement of $H(\tau_0)$ from the initial configuration $\tau_0$ by following the sequence of moves in $L$. Then, $h(L)\le n^2$ since $|X_e|\le 1$ for every $e\in E$. Let $h(L_i)$ denote the total improvement of $H(\tau_i)$ from the initial configuration $\tau_i$ by following the sequence of moves in $L_i$. Then, $h(L)=\sum_{i=1}^z h(L_i)$. Hence, there exists $i\in [z]$ such that 
\[
h(L_i)\le \frac{n^2}{z}\le e^{-\frac{2(1+2\beta)}{\beta}}\phi^{-1}n^{-\left(\frac{1+2\beta+2\beta^2}{\beta}+\frac{\eta(1+2\beta)}{\beta}\right)}.
\]
The second inequality above is by the lower bound on $z$ from \eqref{eq:z-lower-bound-cut}. 
Thus, there exists $i\in [z]$ such that $L_i$ is $\epsilon$-slowly improving from $\tau_i$ with respect to $X$, where $\epsilon:=e^{-\frac{2(1+2\beta)}{\beta}}\phi^{-1}n^{-\left(\frac{1+2\beta+2\beta^2}{\beta}+\frac{\eta(1+2\beta)}{\beta}\right)}$. 

The above argument implies that if $\mathcal{R}_X$ happens, then there exists a sequence $L'$ of moves of length $\lceil (1+\beta)n\rceil$ and an initial configuration $\sigma_0\in \{\pm 1\}^V$ such that $L'$ is $\epsilon$-slowly improving from $\sigma_0$ with respect to $X$. By Lemma \ref{lemma:epsimp-2}, the probability of the latter event is $o(1)$ and hence the probability that $\mathcal{R}_X$ happens is $o(1)$. 

It remains to identify a setting of $\beta$ that bounds the run-time. That is, we need a setting of $\beta$ that minimizes the exponent of $n$ in the RHS of \eqref{eq:length-expression-cut}. The optimal choice is $\beta=1/\sqrt{2}$. Thus, the probability that an implementation of FLIP starting from some initial configuration $\tau_0$ follows a sequence $L$ of length at least 
\begin{equation*}
1580\phi n^{(2+\sqrt{2})(\sqrt{2}+\eta)}.
\end{equation*}
is $o(1)$. 

\end{proof}

\section{Smoothed analysis of FLIP for \kcut}
In this section we prove Theorems \ref{theorem:max-3-cut-complete-poly-time} and \ref{theorem:max-k-cut-quasi-poly}.
We begin with some notations. Let $G=(V,E)$ be an \emph{arbitrary} connected graph with $n$ vertices and let $X:E\rightarrow [-1,1]$ be an edge-weight function. We will redefine some of the concepts from Section \ref{sec:2cut} as there are subtle differences between the same notions between the case of \cut and \kcut. For the sake of completeness, we state the complete definition and prove all necessary details.

We recall a convenient formulation of the objective function for \kcut \cite{FJ95}.
When considering \kcut, let $\sigma(1), \ldots, \sigma(k)$ be vectors defined as follows:
take an equilateral simplex $\Sigma_k$ in $\mathbb{R}^{k-1}$ with vertices $b_1, \ldots, b_k$.
Let $c_k := (b_1 + \cdots + b_k)/k$ be the centroid of $\Sigma_k$ and let $\sigma(i) = b_i - c_k$, for $i \in [k]$.
Assume that $\Sigma_k$ is scaled such that $|\sigma(i)| = 1$ for $i \in [k]$.
For example, \threecut produces the vectors:
\begin{align*}
\sigma(1) := \frac{1}{\sqrt{6}} (-2,1,1), \; \sigma(2) := \frac{1}{\sqrt{6}} (1,-2,1), \text{ and } &\sigma(3) := \frac{1}{\sqrt{6}} (1,1,-2).
\end{align*}

\begin{remark}
If $i, j \in [k]$, then
\[
\langle \sigma(i), \sigma(j) \rangle =
\begin{cases}
1 & \text{ if } i = j,\\
\frac{-1}{k-1} & \text{ if } i \not = j.
\end{cases}
\]
\end{remark}

We consider the space $[k]^V$ of configurations that define a partition of the vertex set into $k$ parts.
For a configuration $\tau\in [k]^V$, we denote the part of $v$ by $\tau (v)$. For a configuration $\tau\in [k]^V$, the weight of $\tau$ is given by
\[
\frac{k-1}{k}\sum_{uv \in E} X(uv) (1 - \langle\sigma(\tau(u)),\sigma(\tau(v))\rangle).
\]
Let
\[
H(\tau) := -\frac{k-1}{k} \sum_{uv \in E} X_{uv}\langle\sigma(\tau(u)),\sigma(\tau(v))\rangle.
\]
We observe that for constant $k$, $H(\tau)$ is a translation of the weight of $\tau$ by some fraction of the total weight of all edges and hence, it suffices to work with $H(\tau)$ henceforth.

We analyze the run-time of the FLIP method in the smoothed framework.
We will denote the \emph{move} of a vertex $v\in V$ from part $p\in [k]$ to part $q\in [k]\setminus \{p\}$ as an ordered triple $(v,p,q)$. A move $(v,p,q)$ is \emph{valid} for a configuration $\tau\in [k]^V$ if $\tau(v)=p$ and $q\neq p$.
We will need the notions of valid and improving sequences that we define now.

\begin{definition}
Let $L$ be a sequence of moves, $\tau_0\in [k]^V$ be an initial configuration, and $X\in [-1,1]^E$ be the edge weights.
We will denote the length of the sequence $L$ by $\ell(L)$, the set of vertices appearing in the moves in $L$ by $S(L)$, and $s(L):=|S(L)|$. For each $v\in V$, we will denote the number of times that the vertex $v$ moves in $L$ by $\#_L(v)$.
We will denote the $t$'th move of $L$ by $L(t)=(v_t,p_t,q_t)$.
\begin{enumerate}
\item For each $t\in [\ell(L)]$ such that $L(t)$ is valid for $\tau_{t-1}$, we will denote $\tau_t$ as the configuration obtained from $\tau_{t-1}$ by setting $\tau_t(u):=\tau_{t-1}(u)$ for every $u\in V\setminus \{v_t\}$ and $\tau_t(v_t):=q_t$. If there exists $t\in [\ell(L)]$ such that $L(t)$ is invalid for $\tau_{t-1}$, then we say that $L$ is invalid from $\tau_0$; otherwise $L$ is \emph{valid} from $\tau_0$.
\item We say that \emph{$L$ is improving from $\tau_0$ with respect to $X$} if $L$ is valid from $\tau_0$ and $H(\tau_{t})-H(\tau_{t-1})>0$ for all $t\in [\ell(L)]$. We say that \emph{$L$ is $\epsilon$-slowly improving from $\tau_0$ with respect to $X$} if $L$ is valid from $\tau_0$ and $H(\tau_{t})-H(\tau_{t-1})\in (0,\epsilon]$ for all $t\in [\ell(L)]$.
\end{enumerate}
\end{definition}
The notion of valid sequences is needed only for $k\ge 3$ in the case of \kcut and was not necessary for \cut in the previous section. Moreover, we emphasize that the definition of $\epsilon$-slowly improving here is different from the one that we used in Section \ref{sec:2cut} for \cut.
Next, we obtain a convenient expression for characterizing the improvement of $H(\tau)$ in each step.
\begin{definition}
Let $L$ be a valid sequence of moves from a configuration $\tau_0\in [k]^V$.
Let $M_{L,\tau_0}\in \{0,\pm 1\}^{E\times \ell(L)}$ be a matrix with rows corresponding to the edges of $G$, columns corresponding to time-steps in the sequence $L$, and whose entries are given by
\[
M_{L,\tau_0}[\{a,b\},t]:=
\begin{cases}
+1 &\mbox{if } a=v_t \text{ and } q_t = \tau_t(b), \text{ or } b=v_t \text{ and } q_t = \tau_t(a),\\
-1 &\mbox{if } a=v_t \text { and } p_t = \tau_t(b), \text{ or } b=v_t \text { and } p_t = \tau_t(a),\\
0 &\mbox{otherwise}, 
\end{cases}
\]
where $\{a,b\}\in E$ and $t\in [\ell(L)]$. We will denote the $t$'th column of $M_{L,\tau_0}$ by $M_{L,\tau_0}^t$.
\end{definition}
\begin{remark}
For a sequence $L$ that is valid from an initial configuration $\tau_0$, we have $H(\tau_t)-H(\tau_{t-1})=\langle M_{L,\tau_0}^t,X\rangle$.
\end{remark}

Next, we need the notion of cycles and cyclic vertices. We note that the following definitions do not depend on the initial configuration.
\begin{definition}
Let $L$ be a sequence of moves.
\begin{enumerate}
%

\item
A set of $w$ moves $\{(v_{t_1},p_{t_1},q_{t_1}), \ldots, (v_{t_w}, p_{t_w}, q_{t_w})\}$ in $L$ is a \emph{$w$-circuit} over a vertex $v \in S(L)$ if
\begin{enumerate}
\item $t_i < t_j$ for all $i < j$,
\item $q_{t_i} = p_{t_{i+1}}$ for all $i \in [w-1]$,
\item $q_{t_w} = p_{t_1}$ and,
\item $v_{t_i} = v$ for all $i \in [w]$.
\end{enumerate}
We will denote the time steps $\{t_1, \ldots, t_w\}$ of the $w$-circuit by $T(C)$.

\item A $w$-circuit is a \emph{$w$-cycle} if it is inclusion-wise minimal.

\item A set $C$ of moves in $L$ is a \emph{cycle} if it is a $w$-cycle for some $w$.
Also, let $t_{beg}(C) :=\min(T(C))$ and $t_{end}(C) := \max(T(C))$. Let $\Gamma(L)$ denote the set of all cycles in $L$.

\item A vertex $v$ is called \emph{cyclic} if there exists a cycle in $\Gamma(L)$ that is over $v$. A vertex $v$ is called \emph{acyclic} if it is not cyclic. Let $C(L)$ and $A(L)$ denote the set of cyclic and acyclic vertices of $L$ respectively, and let $c(L) := |C(L)|$ and $a(L) := |A(L)|$.
\end{enumerate}
\end{definition}

\begin{remark}
Let $L$ be a sequence of moves and let $C\in \Gamma(L)$ be a
cycle over a vertex $v$. Then, every part is visited at most once by $v$ in the cycle.
\end{remark}

\begin{remark}
\label{remark:acyclicmax}
For a sequence $L$ of moves,
we have that $\#_L(v)\le k-1$ for each vertex $v\in A(L)$.
\end{remark}

We now define a suitable matrix that will nullify the influence of non-moving vertices.
\begin{definition}
Let $L$ be a sequence of moves that is valid from an initial configuration $\tau_0\in [k]^V$. Let $P_{L,\tau_0} \in  \{0,\pm 1\}^{E\times \Gamma(L)}$ be a matrix with rows corresponding to edges of $G$, columns corresponding to cycles in $L$, and whose entries are given by
\[
P_{L,\tau_0}[\{a,b\}, C] := \sum\limits_{(v_t, p_t, q_t) \in C} M[\{a,b\}, t],
\]
where $\{a,b\}\in E$ and $C\in \Gamma(L)$.
\end{definition}

\begin{prop}\label{prop:non-moving-nullified}
For a sequence $L$ of moves that is valid from an initial configuration $\tau_0 \in [k]^V$, if $v\in V\setminus S(L)$, then $P_{L,\tau_0}[\{a,v\},C]=0$ for every $C\in \Gamma(L)$ and $\{a,v\}\in E$.
\end{prop}
\begin{proof}
Let $C \in \Gamma(L)$ and $\{a,v\}\in E$.
Since $v$ is not in $S(L)$, it follows that $C$ is not over $v$. If $C$ is not over $a$, then $M_{L,\tau_0}[\{a,v\},t]=0$ for every $t\in T(C)$ and hence $P_{L,\tau_0}[\{a,v\},C]=0$.
Suppose $C$ is a $w$-cycle over the vertex $a$ and let $C = \{(a, p_{t_1}, q_{t_1}), \ldots, (a, p_{t_w}, q_{t_w}) \}$.
If $M_{L,\tau_0}[\{a,v\},t_i] = 0$ for all $i \in [w]$, then the claim holds.
So, without loss of generality we assume that $M_{L,\tau_0}[\{a,v\},t_1] = 1$.
Then $M_{L,\tau_0}[\{a,v\},t_i] = 0$ for all $1 < i < w$ since $q_{t_i} \not = p_{t_1}$ for all $i \not = w$.
Finally, it follows that $M_{L,\tau_0}[\{a,v\},t_w] = -1$ since $v$ does not move between $t_1$ and $t_w$ in $L$. Hence, $P_{L,\tau_0}[\{a,v\},C]=\sum_{i=1}^w M_{L,\tau_0}[\{a,v\},t_i]=0$.

%
\end{proof}

\subsection{Rank lower bounds for $P_{L,\tau_0}$}
In this section, we show a lower bound on the rank of $P_{L,\tau_0}$.
For this, we will use a directed graph with certain properties. We define these properties now.

\begin{definition} \label{definition:good-ngbrwise-indep}
Let $L$ be a sequence of moves.
\begin{enumerate}[(i)]
\item \label{item:good}
For $u,v\in S(L)$, we will call the ordered pair $uv$ to be an
\emph{$L$-good-arc} if there exists a cycle $C\in \Gamma(L)$ over $u$ such that $P_{L,\tau_0}[\{u,v\},C]\neq 0$. A directed graph $H$ whose nodes are a subset of $S(L)$, is \emph{$L$-good} if every arc in $H$ is an $L$-good-arc.

\item \label{item:nghbrwise-indep} For a cyclic vertex $v\in C(L)$ and a collection $U\subseteq S(L)$ of vertices with $m:=|U|$, the collection of ordered pairs $\{vu:u\in U\}$ is an $L$-neighbor-wise independent set if there exists an ordering of $U$, say $u_1,\ldots, u_m$, along with cycles $C_1,\ldots, C_m\in \Gamma(L)$ over $v$ such that
\begin{enumerate}
\item $P_{L,\tau_0}[\{v,u_i\},C_i]\neq 0$ and
\item $P_{L,\tau_0}[\{v,u_j\},C_i]= 0$ for every $j\in \{i+1,\ldots, m\}$.
\end{enumerate}
A directed graph $H$ with node set $S(L)$ is $L$-neighbor-wise independent if for every $v\in S(L)$, the collection $\{vu: u\in \Delta^{out}_H(v)\}$ is an $L$-neighbor-wise independent set.

\end{enumerate}
\end{definition}

The next lemma is our key tool in obtaining a lower bound on the rank of $P_{L,\tau_0}$. We show that the rank is at least the number of edges in an $L$-good $L$-neighbor-wise independent directed acyclic graph. The proof of this lemma is identical to the proof of Lemma \ref{lemma:hrank-cut} in Section \ref{sec:2cut}. We include its proof for the sake of completeness since our definitions have changed mildly.
\begin{lemma}\label{lemma:3hrank}
Let $L$ be a valid sequence from an initial configuration $\tau_0 \in [k]^V$.
Let $H$ be an $L$-good $L$-neighbor-wise independent directed acyclic graph.
Then,
\[
rank(P_{L, \tau_0}) \geq |E(H)|.
\]
\end{lemma}
\begin{proof}
Consider the submatrix $B_H$ of $P_{L,\tau_0}$ consisting of the rows corresponding to edges $\{u,v\}$ for every arc $vu\in E(H)$. We will show that the matrix $B_H$ has full row-rank by induction on $|E(H)|$. The base case of $|E(H)|=0$ is trivial.

For the induction step, we consider $|E(H)|\ge 1$. Suppose that there exist coefficients $\mu_{\{u,v\}}\in R$ for every $uv\in E(H)$ such that
\[
\sum_{uv\in E(H)} \mu_{\{u,v\}}P_{L,\tau_0}[\{u,v\},C]=0
\]
for every cycle $C\in \Gamma(L)$. Since $H$ is a directed acyclic graph with at least one arc, there exists a node $v_s\in V(H)$ with $|\delta^{out}_H(v)|\ge 1$ and $|\delta_H^{in}(v)|=0$.
\begin{claim}
For every $u\in \Delta_H^{out}(v)$, the coefficient $\mu_{\{v_s,u\}}$ is zero.
\end{claim}
\begin{proof}
Consider the ordering $u_1,\ldots, u_m$ of the vertices in $\Delta^{out}_H(v)$ and cycles $C_i \in \Gamma(L)$ satisfying Definition \ref{definition:good-ngbrwise-indep} (\ref{item:nghbrwise-indep}). We show that $\mu_{\{v_s,u_j\}}=0$ for every $j\in[m]$ by induction on $j$.

For the base case, we consider $j=1$. Consider the column of $P$ corresponding to the cycle $C_1$. Since $C_1$ is over $v$, the only possible non-zero entries in this column among the chosen rows are in the rows corresponding to the edges $\{v_s,u_1\},\ldots, \{v_s,u_m\}$. Thus,
\begin{align*}
0
&=\sum_{uv\in E(H)}\mu_{\{u,v\}}P_{L,\tau_0}[{\{u,v\},C_1}]
=\sum_{i=1}^m\mu_{\{v_{s},u_i\}}P_{L,\tau_0}[{\{v_s,u_i\},C_1}]. \label{eq:base-case-1}
\end{align*}
Moreover, by the choice of $C_1$, we have that
\begin{align*}
P_{L,\tau_0}[\{v_s,u_1\},C_1] &\neq 0 \text{, and}\\
P_{L,\tau_0}[\{v_s,u_i\}, C_1] &= 0 \ \forall i \in [m] \setminus \{1\}.
\end{align*}
Consequently, we obtain that $\mu_{\{v_s,u_1\}}=0$.
For the induction step, consider $j\ge 2$. Consider the columns of $P$ corresponding to the cycle $C_j \in \Gamma(L)$. Since $C_j$ is also over $v$, the only possible non-zero entries in this column among the chosen rows are in the rows corresponding to the edges $\{v_s,u_1\},\ldots, \{v_s,u_m\}$. Thus,
\begin{align*}
0
&=\sum_{uv\in E(H)}\mu_{\{u,v\}}P_{L,\tau_0}[{\{u,v\},C_j}]
=\sum_{i=1}^m\mu_{\{v_{s},u_i\}}P_{L,\tau_0}[{\{v_s,u_i\},C_j}].
\end{align*}
By induction hypothesis, we know that $\mu_{\{v_s,u_i\}}=0$ for every $i\in\{1,2,\ldots, j-1\}$. Thus,
\begin{align*}
0 &= \sum_{i=j}^m\mu_{\{v_{s},u_i\}}P_{L,\tau_0}[{\{v_s,u_i\},C_j}].
\end{align*}
Moreover, by the choice of $C_j$, we also have that
\begin{align*}
P_{L,\tau_0}[\{v_s,u_1\},C_j] &\neq 0 \text{, and}\\
P_{L,\tau_0}[\{v_s,u_i\}, C_j] &= 0 \ \forall i \in \{j+1,\ldots,m\}.
\end{align*}
Consequently, we obtain that $\mu_{\{v_s,u_j\}}=0$.

\end{proof}

As a consequence of the claim, we have that the matrix $B_H$ has full row-rank if and only if the matrix $B_{H'}$ obtained from the graph $H':=H-\delta_H^{out}(v_s)$ has full row-rank. We note that the graph $H'$ is also an $L$-good $L$-neighbor-wise independent directed acyclic graph with $|E(H')|<|E(H)|$. Thus, by induction hypothesis, the matrix $B_{H'}$ has full row-rank.
Hence, the matrix $B_H$ also has full row-rank.
\end{proof}

We now use Lemma \ref{lemma:3hrank} to show that the rank of the matrix $P_{L,\tau_0}$ is at least half the number of cyclic vertices in $L$ provided that $L$ is an improving sequence from some initial configuration.

\begin{lemma} \label{lemma:half}
Let $L$ be an improving sequence from an initial configuration $\tau_0 \in [k]^V$ with respect to some edge weights $X\in [-1,1]^E$. Then,
\[
rank(P_{L, \tau_0}) \geq \frac{1}{2}c(L).
\]
\end{lemma}

\begin{proof}
We will use the following claim to construct an $L$-good $L$-neighbor-wise independent directed acyclic graph.
\begin{claim}
For every vertex $v\in C(L)$ and for every cycle $C\in \Gamma(L)$ that is over $v$, there exists an edge $\{u,v\} \in E$ such that $P_{L,\tau_0}[\{u,v\}, C] \neq 0$.
\end{claim}
\begin{proof}
For contradiction, suppose that for all edges $\{u,v\}\in E$ we have $P_{L,\tau_0}[\{u,v\}, C] = 0$.
We note that for all $e \in E$, if $v$ is not an end-vertex of $e$, then $M_{L,\tau_0}[e, t]=0$ for every $t\in T(C)$. Moreover, for all $e \in E$ with $v$ being an end-vertex of $e$, we have that
\[
 \sum_{t \in T(C)} M_{L,\tau_0}[e,t] = P_{L,\tau_0}[e, C]  = 0.
\]
Hence, for every $e\in E$, we have that
\[
\sum_{t \in T(C)} M_{L,\tau_0}[e,t]= 0.
\]
This implies that $\sum_{t \in T(C)} \langle M_{L,\tau_0}^t, Y\rangle = 0$ for all $Y \in [-1,1]^E$.
However, since $L$ is an improving sequence from $\tau_0$ with respect to $X \in [-1,1]^E$, it follows that $\langle M_{L,\tau_0}^t, X\rangle > 0$ for all $t \in [\ell(L)]$.
In particular, $\sum_{t \in T(C)} \langle M_{L,\tau_0}^t, X\rangle > 0$, a contradiction.
\end{proof}

Now we construct an $L$-good $L$-neighbor-wise independent graph $H$ over the node set $S(L)$ as follows: For every $v \in C(L)$, pick an arbitrary $u \in V \setminus \{v\}$ such that $P_{L,\tau_0}[\{u,v\}, C] \neq 0$ (which is guaranteed to exist by the above claim) and add the arc $vu$
to $H$. The resulting graph $H$ is $L$-good by construction. It is trivially $L$-neighbor-wise independent since each node has out-degree at most one. Moreover, $|E(H)|=c(L)$ and the directed cycles in $H$ are node-disjoint.

Finally, we obtain an $L$-good $L$-neighbor-wise independent directed acyclic graph $H'$ by removing one arc from each directed cycle in $H$. Since $|E(H)|=c(L)$ and the directed cycles in $H$ are node-disjoint, it follows that $|E(H')| \geq \frac{1}{2}c(L)$. The theorem now follows by applying Lemma \ref{lemma:3hrank} to $H'$.
\end{proof}

Next, we will improve the rank lower bound from Lemma \ref{lemma:half} in complete graphs for \threecut. We need a few additional definitions. The following definition will also be useful for the quasi-polynomial time analysis for \kcut.
\begin{definition}
Let $L$ be a sequence of moves. 
A \emph{block} is a continuous subsequence of $L$. For a block $L'$ of $L$, we will denote the set of time-steps of the moves of $L'$ in $L$ by $T(L')$.
A maximal block of $L$ consisting only of cyclic vertices of $L$ is called a cyclic block.
Likewise, a maximal block of $L$ consisting only of acyclic vertices of $L$ is called an acyclic block.
\end{definition}

We note that a sequence $L$ can be partitioned into alternating cyclic and acyclic blocks.

\subsubsection{Improving rank lower bounds for \threecut in the complete graph}
We now focus on the case when $k = 3$ and $G$ is the complete graph. For $k=3$, the cycles of interest are $2$-cycles and $3$-cycles.

\begin{definition}
A cycle $C \in \Gamma(L)$ over a vertex $v \in V$ is \emph{leaping} if the time-steps in $T(C)$
belong to at least two distinct cyclic blocks of $L$.
A leaping $3$-cycle $\{(v_{t}, p_{t}, q_{t}),(v_{t'}, p_{t'}, q_{t'}),(v_{t''}, p_{t''}, q_{t''})\} \in \Gamma(L)$ is called \emph{tricky} if
\begin{enumerate}
\item the time-steps $t,t',t''$ belong to distinct cyclic blocks of $L$ and
\item the set of acyclic vertices of $L$ which appear between $t$ and $t'$ is the same as those which appear between $t'$ and $t''$ (i.e., $A(L) \cap S(L[t,t']) = A(L) \cap S(L[t',t''])$).
\end{enumerate}
Here, $L[a,b]$ denotes the subsequence of moves that occur between time-step $a$ and time-step $b$ 
(inclusive of the boundaries).
\end{definition}

The following two lemmas summarize the structure of cyclic blocks and leaping cycles.

\begin{lemma} \label{lemma:cycle3}
Let $L$ be a valid sequence from some initial configuration $\tau_0 \in [3]^V$ and let $t_1 < t_2 < t_3$ be the time-steps of three occurrences of a vertex $v$ in $L$ such that $t_1,t_2,t_3$ belong to different cyclic blocks.
Then, there exists a cycle $C \in \Gamma(L)$ over $v$ such that
\begin{enumerate}[(i)]
\item $C$ is a leaping cycle and
\item $t\in \{t_1,t_1+1,\ldots, t_3'\}$ for every $t\in T(C)$, where $t_3'$ is the last occurrence of $v$ in the same cyclic block as that of $t_3$.
\end{enumerate}
\end{lemma}
\begin{proof}
Let the time-steps $t_1,t_2$ and $t_3$ be in cyclic blocks $B_1,B_2$ and $B_3$ respectively.
Let $t'_1,t'_2$ and $t'_3$ denote the last occurrence of $v$ in $B_1,B_2$ and $B_3$ respectively.
Without loss of generality, suppose that the vertex $v$ moves from part $1$ to part $2$ at time-step $t_1'$ (i.e., $p_{t'_1} = 1$ and $q_{t'_1} = 2$).
We consider the following three cases:

\begin{itemize}
\item {\bf Case 1.} There exists a time-step $t$ such that $t'_1 \leq t \leq t'_3$ and $L(t) = (v, 2, 1)$. Then, the cycle $C \in \Gamma(L)$ given by $\{(v,1,2),(v,2,1)\}$ with $T(C) = \{t'_1, t\}$ is the desired cycle.

\item {\bf Case 2.} There does not exist a time-step $t$ such that $t'_1 \leq t \leq t'_3$ and $L(t) = (v, 2, 1)$, but there exists a time-step $t$ such that $t'_1 \leq t \leq t'_3$ and $L(t) = (v, 3, 1)$.
Since $L$ is a valid sequence from an initial configuration $\tau_0$, there exists a time-step $t'$ such that 
$t'_1 \leq t' \leq t$ and $L(t')= (v, 2, 3)$.
Hence, the cycle $C \in \Gamma(L)$ given by $\{(v,1,2), (v,2,3),(v,3,1)\}$ with $T(C) = \{t'_1, t', t\}$ is the desired cycle.

\item {\bf Case 3.} There does not exist a time-step $t$ such that $t'_1 \leq t \leq t'_3$ and $L(t) = (v, 2, 1)$, and there does not exist a time-step $t$ such that $t'_1 \leq t \leq t'_3$ and $L(t) = (v, 3, 1)$. Let $t$ be the first occurrence of $v$ in $B_3$.
Since $L$ is a valid sequence for an initial configuration $\tau_0$, we have that $\{L(t'_2),L(t)\} = \{(v,2,3),(v,3,2)\}$.
Hence, the cycle $C \in \Gamma(L)$ given by $\{(v,2,3),(v,3,2)\}$ with $T(C) = \{t'_2,t\}$ is the desired cycle.
\end{itemize}
\end{proof}

\begin{lemma} \label{lemma:nontricky}
Let $G=(V,E)$ be the complete graph and let $L$ be a valid sequence from some initial configuration $\tau_0 \in [3]^V$. Suppose $C \in \Gamma(L)$ is a non-tricky leaping cycle over a vertex $v$.
Then there exists a vertex $u \in A(L) \cap S(L[t_{beg}(C),t_{end}(C)])$ such that the arc $vu$ is an $L$-good-arc.
\end{lemma}
\begin{proof}
We consider the following three cases:
\begin{itemize}
\item {\bf Case 1.} Suppose that $C$ is a $2$-cycle and $T(C)=\{t_1,t_2\}$.
Without loss of generality, suppose that $C=\{ (v, 1,2), (v,2,1)\}$.
Since $C$ is a leaping cycle, there exists a vertex $u \in A(L) \cap S(L[t_{beg}(C), t_{end}(C)])$.
Let $p,q \in [3]$ denote the parts of $u$ at time-steps $t_1$ and $t_2$ respectively.
Since $u \in A(L)$, we have that $p \neq q$.
Therefore, $\{p,q\} \in \left\{ \{1,2\} , \{1,3\}, \{2,3\}\right\}$.
Hence, $P_{L,\tau_0}[\{u,v\},C] \in \{\pm 1,\pm 2\}$ and it is non-zero.

\item {\bf Case 2.} Suppose that $C$ is a 3-cycle and the time-steps in $T(C) = \{t_1, t_2, t_3\}$ belong to exactly two distinct cyclic blocks of $L$.
Without loss of generality, suppose that $C = \{(v,1,2), (v,2,3), (v,3,1) \}$
and that $t_1,t_2$ belong to the same cyclic block which is different from that of $t_3$.
We note that there exists $u \in A(L) \cap S(L[t_2,t_3])$.
Let $p,q \in [3]$ denote the parts of $u$ at time-steps $t_2$ and $t_3$ respectively.
Since $u \in A(L)$, we have that $p \neq q$.
Therefore, $\{p,q\} \in \left\{ \{1,2\} , \{1,3\}, \{2,3\}\right\}$.
Hence, $P_{L,\tau_0}[\{u,v\},C] \in \{\pm 1,\pm 2\}$ and it is non-zero.

\item {\bf Case 3.} Suppose that $C$ is a 3-cycle and the time-steps in $T(C) = \{t_1, t_2, t_3\}$ belong to exactly three distinct cyclic blocks of $L$.
Without loss of generality, suppose that $C = \{(v,1,2), (v,2,3), (v,3,1)\}$.
Since $C$ is non-tricky, there exists $u \in A(L)$ such that $u$ is in exactly one of $S(L[t_1,t_2])$ and $S(L[t_2,t_3])$.
Without loss of generality, suppose that $u \in S(L[t_2, t_3]) \setminus S(L[t_1, t_2])$.
Let $p,q \in [3]$ denote the position of $u$ at time-steps $t_2$ and $t_3$ respectively.
Since $u \in A(L)$, we have that $p \neq q$.
Therefore, $\{p,q\} \in \left\{ \{1,2\} , \{1,3\}, \{2,3\}\right\}$.
Hence, $P_{L,\tau_0}[\{u,v\},C] \in \{\pm 1,\pm 2\}$ and it is non-zero.
\end{itemize}
\end{proof}



We now show that for every cyclic vertex $v$ in a sequence $L$, we have a large number of $L$-good arcs whose tail is $v$ which also form an $L$-neighbor-wise independent set. This fact will be useful in constructing a large $L$-good $L$-neighbor-wise independent directed acyclic graph which will in turn improve the rank using Lemma \ref{lemma:3hrank}. We emphasize that our proof of this fact will crucially use the fact the graph $G$ is complete.
For a sequence $L$ of moves, we will denote the number of cyclic blocks in which $v$ occurs as $b_L(v)$.

\begin{lemma}\label{lemma:l-good-l-ngbrwise-indep}
Let $G$ be the complete graph, let $L$ be a valid sequence from some initial configuration $\tau_0\in [3]^V$, and let $v\in C(L)$. Then, there exists a collection of $k  \geq  \lceil \frac{1}{2} \lfloor \frac{1}{3}(b_L(v)-1) \rfloor \rceil$ vertices $u_1,\ldots, u_k$ such that
\begin{enumerate}[(i)]
\item $u_1,\ldots, u_{k}\in A(L)$,\label{cond:l-good-l-ngbrwise-indep:1}
\item $vu_i$ is an $L$-good-arc for every $i\in [k]$, and \label{cond:l-good-l-ngbrwise-indep:2}
\item the set $\{vu_1,\ldots, vu_{k}\}$ is an $L$-neighbor-wise independent set. \label{cond:l-good-l-ngbrwise-indep:3}
\end{enumerate}
\end{lemma}
\begin{proof}


Let the cyclic blocks that contain $v$ be $B_1,\dots,B_{b_L(v)}$.
Let $R := \lfloor \frac{1}{3} (b_L(v) - 1)\rfloor$.
For $0 \leq r \leq R$, let $\mathbf{B}_r := \{B_{3r+1},B_{3r+2},B_{3r+3},B_{3r+4}\}$. We note that the last block in $\mathbf{B}_r$ and the first block in $\mathbf{B}_{r+1}$ coincide for every $0 \leq r < R$.
Let $\mathbf{A}_r$ be the set of three acyclic blocks in $L$ that appear between $B_{3r+1}$ and $B_{3r+2}$, between $B_{3r+2}$ and $B_{3r+3}$ and between $B_{3r+3}$ and $B_{3r+4}$.
For simplicity, we let $S(\mathbf{A}_r) := \cup_{A \in \mathbf{A}_r}{S(A)}$ denote the set
of vertices occurring in the blocks in $\mathbf{A}_r$, and
let $T(\mathbf{A}_r) := \cup_{A \in \mathbf{A}_r} T(A)$ denote the set of time-steps of the moves of blocks in $\mathbf{A}_r$ as they appear in $L$.
By definition, we have that $T(\mathbf{A}_r) \cap T(\mathbf{A}_{r'}) = \emptyset$ for every distinct $r,r' \in [R]$.

Let us consider an $r\in [R]$.
By Lemma \ref{lemma:cycle3}, there exist leaping cycles $C_r$ and $C'_r$ in $\Gamma(L)$ over $v$ such that $T(C_r) \subseteq T(B_{3r+1}) \cup T(B_{3r+2}) \cup T(B_{3r+3})$ and $T(C'_r) \subseteq T(B_{3r+2}) \cup T(B_{3r+3}) \cup T(B_{3r+4})$.

\begin{claim} \label{claim:both-tricky}
Both $C_r$ and $C'_r$ cannot be tricky cycles.
\end{claim}
\begin{proof}
For the sake of contradiction suppose that both $C_r$ and $C'_r$ are tricky cycles.
Let $T(C_r)=\{t_1,t_2,t_3\}$ and $T(C'_r)=\{t'_1,t'_2,t'_3\}$
for some $t_1 < t_2 < t_3$ and $t'_1 < t'_2 < t'_3$.
By the choice of $C_r$ and $C_r'$, we must have $t_i \in T(B_{3r+i})$ and $t'_i \in T(B_{3r+i+1})$ for $i \in \{1,2,3\}$.
Therefore, $A(L) \cap S(L[t_2,t_3]) = A(L) \cap S(L[t'_1,t'_2])$ and since $C_r$ and $C'_r$ are tricky, it follows that $A(L) \cap S(L[t_1,t_2]) = A(L) \cap S(L[t_2,t_3])$ and $A(L) \cap S(L[t'_1,t'_2]) = A(L) \cap S( L[t'_2,t'_3])$.
Hence, $A(L) \cap S(L[t_1,t_2]) = A(L) \cap S(L[t_2,t_3]) = A(L) \cap S(L[t'_2,t'_3])$.
Since neither of these sets are empty, there must exist a $u \in A(L)$ which appears at least $3$ times in the sequence.
However, this contradicts the fact that for all $u \in A(L)$, we have $\#_L(u) \leq 2$.
\end{proof}

Let $\hat{C}_r := C_r$ if $C_r$ is non-tricky, and $\hat{C}_r := C'_r$ otherwise. It follows from Claim \ref{claim:both-tricky} that
$\hat{C}_r$ is a non-tricky cycle.
Therefore, by Lemma \ref{lemma:nontricky}, there exists a vertex $u_r \in A(L)$ which appears between two cyclic blocks in $\mathbf{B}_r$ such that $vu_r$ is an $L$-good-arc.

We have shown that for each $r \in [R]$, there exists a vertex $u_r\in S(\mathbf{A}_r)$ which appears between two cyclic blocks in $\mathbf{B}_r$ such that $vu_r$ is an $L$-good-arc. Let $U := \{u_1, \ldots, u_R\}$.
We note that the vertices $u_1,\dots,u_R$ may not be distinct and consequently, we may not be able to obtain a large number of $L$-good-arcs while constructing the needed $L$-good $L$-neighbor-wise independent directed acyclic graph. Even if they are distinct, we need an ordering of them that satisfies the $L$-neighbor-wise independent property.
We handle these two issues next.

\begin{claim} \label{claim:halfclaim} 
There exists a subset of $k \ge R/2$ distinct elements in $U$ along with an ordering $u_{r_1}, \ldots, u_{r_{k}}$ of these elements such that $u_{r_j}\not\in S(\mathbf{A}_{r_i})$ for every $i,j\in [k]$ with $i<j$.
\end{claim}

\begin{proof}
We will construct a sequence $w_1,\dots, w_{k}$ such that
\begin{enumerate}[(i)]
\item $k \geq R/2$,
\item $w_i \in U$ for every $i \in [k]$,
\item $w_i \neq w_j$ for every $1 \leq i < j \leq k$,
\item for every $1 \leq i < j \leq k$, we have that $w_i \not \in S(\mathbf{A}_t)$, where $t$ is any index in $[R]$ such that $w_j=u_t$. 
\end{enumerate}

We show that such a sequence translates into the sequence $u_{r_1},\dots,u_{r_{k}}$ required in the claim.
Conditions (i), (ii), and (iii) imply that the elements $w_1,\dots,w_{k}$ are indeed a subset of $k \geq R/2$ distinct elements of $U$.
Let us denote the sequence $w_{k},\ldots, w_1$ by $u_{r_1},\dots,u_{r_{k}}$, where $r_i$ is the
least index of $w_{k-i+1}$ in $u_1,\dots,u_R$.
Substituting $w_i = u_{r_{k-i+1}}$ and $w_j = u_{r_{k-j+1}}$ for all $1 \leq i < j \leq k$ in condition (iv) results in $u_{r_{k-i+1}} \not\in S(\mathbf{A}_{r_{k-j+1}})$ for every $1 \leq i < j \leq k$.
Thus, re-indexing produces
$u_{r_j} \not\in S(\mathbf{A}_{r_{i}})$
for every $i,j\in [k]$ with $i<j$ as desired.

In order to construct a sequence $w_1,\dots, w_k$ satisfying the conditions (i)--(iv), we consider the procedure in Figure \ref{ni-procedure}.
The input satisfies the required conditions as $u_r \in S(\mathbf{A}_r)$ for every $r \in [R]$ by construction of $u_1,\dots,u_R$, and
$u_r$ appears in at most two of the $\mathbf{A}_i$'s since $\#_L(u_r) \leq 2$ and $T(\mathbf{A}_r) \cap T(\mathbf{A}_{r'}) = \emptyset$ for every distinct $r,r'\in [R]$.

We first show that the procedure always terminates.
We note that before the execution of step 2(d), we have $r \in I$ because of step 2(b). Moreover,
from $w_k = u_r \in S(\mathbf{A}_r)$ it follows that $r$ will be removed from $I$ after the execution of step 2(d).
Hence, the size of $I$ decreases by at least one after each iteration of the while loop.

\begin{figure}[ht]
\noindent
\rule{\textwidth}{1pt}
\textbf{Procedure}
\vspace{1mm}
\hrule
\vspace{1mm}
{\bf Input: } Sets $\mathbf{A}_1,\dots,\mathbf{A}_R$ and elements $u_1,\dots,u_R$ with $u_r \in S(\mathbf{A}_r)$ for
every $r \in [R]$, such that
$u_r$ appears in at most two of the $S(\mathbf{A}_i)$'s (i.e. at most one other than $\mathbf{A}_r$).
\begin{enumerate}
\item Initialize $k \leftarrow 0$, $I \leftarrow [R]$.
\item While ($I \neq \emptyset$):
\begin{enumerate}
\item $k\leftarrow k+1$.
\item $r \leftarrow$ an arbitrary element in $I$.
\item $w_{k}\leftarrow u_r$.
\item $I\leftarrow I\setminus \{i\in [R]:w_k\in S(\mathbf{A}_i)\}$.
\end{enumerate}
\item Return $(w_1,\dots,w_k)$.
\end{enumerate}
\rule{\textwidth}{1pt}
\caption{Procedure for Claim \ref{claim:halfclaim}.}
\label{ni-procedure}
\end{figure}

Next, we prove conditions (i)--(iv). Suppose the procedure terminates with the sequence $w_1,\dots,w_k$.

\begin{enumerate}
\item Since each $u_r$ appears in at most two $S(\mathbf{A}_i)$'s, each execution of step 2(d) decreases the size of
$I$ by at most two. Therefore, the while loop iterates at least $R/2$ times.
Hence, $k \geq R/2$.
\item By step 2(c), we have that $w_i \in U$ for every $i \in [k]$.
\item From step 2(b)--2(d) it follows that $w_i \neq w_j$ for $1 \leq i < j \leq k$.
\item
Let $i \in [k]$. Consider the $i$-th iteration of the while loop. After step 2(d) refines the set $I$, we have that $w_i \not \in S(\mathbf{A}_r)$ for all indices $r \in I$ that are chosen in iterations after the $i$'th iteration. Thus, if $t$ is the index of $w_j$ in $u_1,\ldots, u_R$ for some $j\in [k]$ where $j>i$, then $w_i\not\in S(\mathbf{A}_t)$.

\end{enumerate}

\end{proof}

Consider the subset of $k \geq  \lceil \frac{1}{2} \lfloor \frac{1}{3}(b_L(v)-1) \rfloor \rceil$ distinct elements $U' = \{ u_{r_1},\dots,u_{r_{k}}\}  \subseteq U$ from Claim \ref{claim:halfclaim}.
Conditions (\ref{cond:l-good-l-ngbrwise-indep:1}) and (\ref{cond:l-good-l-ngbrwise-indep:2}) desired in the lemma trivially hold since $U' \subseteq U \subseteq A(L)$ and $vu$ is an $L$-good arc for all $u \in U$.
We recall that the cycles $C_{r_1}, \ldots, C_{r_k}$ are disjoint and $P_{L, \tau_0}[\{v,u_{r_i}\}, C_{r_i}] \not = 0$ for all $i \in [k]$.
Then, since $u_{r_i} \not\in S(\mathbf{A}_{r_j})$ for every $j \in \{i+1, \ldots, k\}$, it follows that $P_{L, \tau_0}[\{v,u_{r_i}\}, C_{r_j}] =  0$ for all $j \in\{i+1, \ldots, k\}$.
Hence, Condition (\ref{cond:l-good-l-ngbrwise-indep:3}) holds.
%
\end{proof}

We now have all the tools necessary to show our improved rank lower bound.

\begin{lemma} \label{lemma:rank-improved}
Let $G=(V,E)$ be the complete graph and let $L$ be a valid sequence from some initial configuration $\tau_0 \in [3]^V$. 
Then $rank(P_{L, \tau_0}) \geq \frac{1}{6} \sum_{v \in C(L)} (b_L(v)-3)$.
\end{lemma}
\begin{proof}
For each $v\in C(L)$, let $R(v):=\lceil\frac{1}{2}\lfloor\frac{1}{3}(b_L(v)-1)\rfloor\rceil$.
We construct an $L$-good $L$-neighbor-wise independent directed acyclic graph $H$ over the vertex set $S(L)$ as follows: for each $v\in C(L)$, add the arcs $vu_{r_1},\ldots, vu_{R(v)}$, where $u_{r_1},\ldots, u_{R(v)}$ is the collection of vertices guaranteed to exist by Lemma \ref{lemma:l-good-l-ngbrwise-indep}.
Then it follows that $H$ is $L$-good and $L$-neighbor-wise independent.
Moreover, we note that all arcs in $E(H)$ have heads in $A(L)$ and tails in $C(L)$.
Hence, $H$ does not contain any directed cycles.
Therefore, $H$ is an $L$-good $L$-neighbor-wise independent directed acyclic graph.
For a lower bound on the number of arcs in $H$, we have
\begin{align*}
|E(H)| = \sum_{v \in C(L)} R(v)=\sum_{v \in C(L)} \left\lceil \frac{1}{2} \left\lfloor \frac{1}{3} (b_L(v) - 1) \right\rfloor \right\rceil \geq\frac{1}{6} \sum_{v \in C(L)} (b_L(v) - 3).
\end{align*}

\end{proof}
Finally, we combine the results of Lemma \ref{lemma:half} and Lemma \ref{lemma:rank-improved} to get one final rank bound.
We define a block $B$ in $L$ to be \emph{$2$-critical} if $\ell(B) \geq 3s(B)$ and $\ell(B') < 3s(B')$ for every block $B'$ that is strictly contained in $B$.


\begin{corollary} \label{coro:3crit}
Let $G=(V,E)$ be the complete graph and let $B$ be an improving sequence from some initial configuration $\tau_0 \in [3]^V$ with respect to some edge weights $X\in [-1,1]^E$. Suppose $B$ is a $\beta$-critical block for some $\beta > 0$.
Then, $rank(P_{B,\tau_0}) \geq \frac{1}{32}s(B)$.
\end{corollary}
\begin{proof}
Suppose that B has no acyclic vertices, then 
\[
\text{rank}(P_{B,\tau_0})\ge \frac{1}{2}c(B)=\frac{1}{2}s(B). 
\] 
and we are done. So, we may assume that B has at least one acyclic block. 
Let $B_1,\ldots, B_k$ denote the acyclic blocks of $B$ and for a vertex $v$, let $1_{v\in B_i}$ denote the indicator function for whether the vertex $v$ appears in $B_i$. Then,
\[
\sum_{v\in C(B)}b_B(v) = \sum_{v\in C(B)}\sum_{i=1}^k 1_{v\in B_i}=\sum_{i=1}^k \sum_{v\in C(B)}1_{v\in B_i} = \sum_{i=1}^k s(B_i).
\]
Since each $B_i$ is a proper sub-block of $B$, we have that $\ell(B_i)<3s(B_i)$ and hence,
\[
\sum_{v\in C(B)}b_B(v) = \sum_{i=1}^k s(B_i) \ge \frac{1}{3}\sum_{i=1}^k \ell(B_i).
\]

By Lemma \ref{lemma:rank-improved}, we have that
\begin{align}
rank(P_{B,\tau_0})
&\ge \frac{1}{6} \sum_{v \in C(B)} (b_B(v)-3) \\
&= \frac{1}{6}\left(\sum_{v\in C(B)} b_B(v) - 3c(B)\right)\\
&\geq \frac{1}{6}\left(\frac{1}{3}\sum_{i=1}^k \ell(B_i) - 3c(B)\right) \label{step:bbcrit}\\
&\geq \frac{1}{6} \left(\frac{\ell(B)-2a(B)}{3} - 3c(B)\right) \label{step:acyclic}\\
&\ge \frac{1}{6}\left(s(B)-\frac{2}{3}a(B) - 3c(B)\right) \label{step:bcrit} \\
&= \frac{1}{18} a(B) - \frac{1}{3}c(B) \label{step:sac}
\end{align}
where
step (\ref{step:acyclic}) follows from the fact that acyclic vertices of $B$ appear at most twice in $B$,
step (\ref{step:bcrit})  follows from the fact that $B$ is a critical block,
and step (\ref{step:sac}) follows from $s(B)=a(B)+c(B)$.
Thus, by the above inequality and Lemma \ref{lemma:half}, we have
\[
rank(P_{B, \tau_0})\ge \max\left\{\frac{1}{2}c(B),\frac{1}{18} a(B) - \frac{1}{3}c(B)\right\}.
\]
Let $\lambda:=a(B)/s(B)$.
Then, $a(B)=\lambda s(B)$ and $c(B)=(1-\lambda)s(B)$.
Thus,
\[
rank(P_{B, \tau_0})
\ge \max\left\{\frac{1}{2}(1-\lambda),\frac{1}{18} \lambda - \frac{1}{3} (1 - \lambda) \right\} s(B)
\ge \frac{1}{32}s(B).
\]
\end{proof}

\subsection{Run-time of FLIP for \threecut in the complete graph}
In this subsection we will show that an improving sequence will improve the value of $H(\tau)$ by some non-negligible amount with constant probability. Theorem \ref{theorem:max-3-cut-complete-poly-time} will follow from this result.
First we recall a lemma of \cite{ER17}.

\begin{lemma}\label{lemma:ERprob}\cite{ER17}
Let $\phi > 0$ and $X_1,\dots,X_m$ be independent random variables with
density functions $f_1,\dots,f_m: \R \rightarrow [0,\phi]$, and let
$X:=(x_1,\dots,x_m)^\intercal$.
Let $\alpha_1,\dots,\alpha_k \in \Z^n$ be linearly independent vectors. Then for every $\epsilon>0$,
$$\Pr_X \Big[\langle \alpha_i, X\rangle \in (0,\epsilon] \ \forall\ i \in [k]
\Big] \leq 
(\phi \epsilon)^k.$$
\end{lemma}

We will use the above lemma and the rank lower bound to bound the probability of the existence of a bad starting configuration for critical blocks.
\begin{lemma}\label{lemma:3bbound}
Let $G=(V,E)$ be the complete graph and let $B$ be a $2$-critical block. Then
\[
\Pr_X[\exists\ \tau_0\in [3]^V:\ B \text{ is $\epsilon$-slowly improving from $\tau_0$ with respect to $X$}] \leq 3^{s(B)}
(3\phi \epsilon)^{\frac{1}{32}s(B)}.
\]
\end{lemma}
\begin{proof}
Suppose that $B$ is $\epsilon$-slowly improving from some $\tau_0$ with respect to $X$. Since $B$ is improving from $\tau_0$ with respect to $X$, it follows that $\langle M_{B,\tau_0}^t,X\rangle\in (0,\epsilon]$ for all $t\in [\ell(B)]$. Since every column of $P_{B,\tau_0}$ is the sum of at most three columns of $M_{B,\tau_0}$, we have that $\langle P_{B,\tau_0}^C,X\rangle\in (0,3\epsilon]$ for all $C\in \Gamma(B)$. Hence, the required probability is at most
\[
\Pr_X\Big[\exists \tau_0\in [3]^V: \text{$B$ is valid from $\tau_0$ and $\langle P_{B,\tau_0}^C,X\rangle\in (0,3\epsilon]$ $\forall\ C\in \Gamma(B)$}\Big].\\
\]

Let $\mathcal{I}_{B,\tau_0,X}$ denote the event that $B$ is an improving sequence from $\tau_0$ with respect to $X$. Then, by union bound, the required probability is at most
\[
\sum_{v\in S(B)}\sum_{\tau_0(v)\in [3]} \Pr_X\Big[\exists\ \tau_0(u)\in [3]\ \forall\ u\in V\setminus S(B):\ \mathcal{I}_{B,\tau_0,X} \text{ and } \langle P_{B,\tau_0}^C,X\rangle\in (0,3\epsilon]\ \forall\ C\in \Gamma(B) \Big].
\]

For $\pi_f:S(B)\rightarrow [3]$ and $\pi_c:V\setminus S(B)\rightarrow [3]$, let us define $\tau_{(\pi_f,\pi_c)}:V\rightarrow [3]$ as
\[
\tau_{(\pi_f,\pi_c)}(u):=
\begin{cases}
\pi_f(u) & \text{ if $u\in S(B)$ and}\\
\pi_c(u) & \text{ if $u\in V\setminus S(B)$}.
\end{cases}
\]

We will now bound the following probability for a fixed choice of $\pi_f:S(B)\rightarrow [3]$ and then take a union bound over the choices of the initial configuration for the vertices in $S(B)$: 
\[
\Pr_X\Big[\exists\ \pi_c:V\setminus S(B)\rightarrow [3]:\ \mathcal{I}_{B,\tau_{(\pi_f,\pi_c)},X} \text{ and } \langle P_{B,\tau_{(\pi_f,\pi_c)}}^C,X\rangle\in (0,3\epsilon]\ \forall\ C\in \Gamma(B) \Big].
\]
Let us define an initial configuration
\[
\sigma_0(u):=
\begin{cases}
\pi_f(u)&\text{ if $u\in S(B)$ and }\\
1&\text{ if $u\in V\setminus S(B)$}.
\end{cases}
\]
By Proposition \ref{prop:non-moving-nullified}, we have that $P_{B,\sigma_0}=P_{B,\tau_{(\pi_f, \pi_c)}}$ for all $\pi_c : V \setminus S(B) \rightarrow [3]$. 
Hence,
\begin{align*}
&\Pr_X\Big[\exists\ \pi_c:V\setminus S(B)\rightarrow [3]:\ \mathcal{I}_{B,\tau_{(\pi_f,\pi_c)},X} \text{ and } \langle P_{B,\tau_{(\pi_f,\pi_c)}}^C,X\rangle\in (0,3\epsilon]\ \forall\ C\in \Gamma(B) \Big].\\
&\quad = \Pr_X\Big[\exists\ \pi_c:V\setminus S(B)\rightarrow [3]:\ \mathcal{I}_{B,\tau_{(\pi_f,\pi_c)},X} \text{ and } \langle P_{B,\sigma_0}^C,X\rangle\in (0,3\epsilon]\ \forall\ C\in \Gamma(B) \Big]\\
&\quad = \Pr_X\Big[\langle P_{B,\sigma_0}^C,X\rangle\in (0,3\epsilon]\ \forall\ C\in \Gamma(B) \big|\exists\ \pi_c:V\setminus S(B)\rightarrow [3]:\ \mathcal{I}_{B,\tau_{(\pi_f,\pi_c)},X}\Big]\\
&\quad \quad \quad \quad \times \Pr_X\Big[\exists\ \pi_c:V\setminus S(B)\rightarrow [3]:\ \mathcal{I}_{B,\tau_{(\pi_f,\pi_c)},X}\Big]\\
& \quad \le \Pr_X\Big[\langle P_{B,\sigma_0}^C,X\rangle\in (0,3\epsilon]\ \forall\ C\in \Gamma(B) \big|\exists\ \pi_c:V\setminus S(B)\rightarrow [3]:\ \mathcal{I}_{B,\tau_{(\pi_f,\pi_c)},X}\Big].
\end{align*}
Now, we bound the RHS probability. If there exists $\pi_c:V\setminus S(B)\rightarrow [3]$ such that the sequence $B$ is improving from $\tau_{(\pi_f,\pi_c)}$ with respect to $X$, then by Corollary \ref{coro:3crit}, the rank of $P_{B,\tau_{(\pi_f,\pi_c)}}$ is at least $s(B)/32$. Moreover, we know that $P_{B,\sigma_0}=P_{B,\tau_{(\pi_f,\pi_c)}}$ and hence the rank of $P_{B,\sigma_0}$ is at least $s(B)/32$. Therefore, using Lemma \ref{lemma:ERprob}, the RHS probability is at most 
\[
(3\phi\epsilon)^{\text{rank}\left(P_{B,\sigma_0}\right)}\le (3\phi\epsilon)^{\frac{1}{32}s(B)}.
\]

Hence, the probability required in the lemma is at most
\[
3^{s(B)}(3\phi \epsilon)^{\frac{1}{32}s(B)}
\]
as the number of possible initial configurations for the vertices that move in $B$ is at most $3^{s(B)}$.

\end{proof}

\begin{lemma}\label{lemma:3epsimp}
Let $G$ be the complete graph and let $\epsilon=\phi^{-1}n^{-(96+\eta)}$ for a constant $\eta>0$. Then, the probability (over the choices of $X$) that there exists a sequence $L$ of moves of length $3n$ and an initial configuration $\tau_0\in [3]^V$ such that $L$ is $\epsilon$-slowly improving from $\tau_0$ with respect to $X$ is $o(1)$.
\end{lemma}
\begin{proof}
Let $\mathcal{R}_X$ denote the event that there exists a sequence $L$ of moves of length $3n$ and an initial configuration $\tau_0\in [3]^V$ such that $L$ is $\epsilon$-slowly improving from $\tau_0$ with respect to $X$.
We note that every sequence $L$ of length $3n$ contains a $2$-critical block. Therefore,
if the event $\mathcal{R}_X$ happens, then there exists a $2$-critical block $B$ and an initial configuration $\tau_0\in [3]^V$ such that $B$ is $\epsilon$-slowly improving from $\tau_0$ with respect to $X$. 
Hence,
\begin{align*}
\Pr_X[\mathcal{R}_X]
&\le \sum_{\substack{B:\ B\text{ is critical, }\\ \ell(B)\leq 3n}}\Pr_X\Big[\exists \sigma_0:\ B \text{ is $\epsilon$-slowly improving from $\tau_0$ with respect to $X$}\Big]\\
& \leq \sum_{\substack{B:\ B\text{ is critical, }\\ \ell(B)\leq 3n}} 3^{s(B)} {{(3\phi \epsilon)^{\frac{1}{32}s(B)}}}
\quad \quad \quad \text{(by Lemma \ref{lemma:3bbound})}
\\
& \leq \sum_{s =1}^n n^{3s} 3^{s} {(3\phi \epsilon)^{\frac{1}{32}s}}\\
& \leq \sum_{s =1}^n \left(Cn^{3}\phi^{\frac{1}{32}}\epsilon^{\frac{1}{32}}\right)^s
\end{align*}
for some universal constant $C>0$. Therefore, for $\epsilon=\phi^{-1} n^{-(96 + \eta)}$ the sum tends to 0 as $n \rightarrow \infty$.
\end{proof}


The proof of Theorem \ref{theorem:max-3-cut-complete-poly-time} follows from Lemma \ref{lemma:3epsimp} similar to the proof of Theorem \ref{theorem:max-cut-poly} that follows from Lemma \ref{lemma:epsimp-2}.
In this case, we consider the event that an implementation of FLIP produces a sequence of moves from some initial configuration which has length greater than $\phi n^{99+\eta}$ for any constant $\eta > 0$.
This event implies that there exists a sequence of length $3n$ that is $\epsilon$-slowly improving from some initial configuration, where $\epsilon  = \phi^{-1}n^{-(96 + \eta)}$.
Finally, we note that the probability of such an event is $o(1)$ by Lemma \ref{lemma:3epsimp}.
%

We mention that the run-time analysis using the above techniques can be improved by replacing the notion of $2$-critical blocks with $\beta$-critical blocks and optimizing the value of $\beta$ similar to the ideas in Section \ref{sec:2cut}. This gives a run-time bound of $O(\phi n^{90.81+\eta})$. We avoid writing a proof of this improved bound in the interests of simplicity.

\subsection{Run-time of FLIP for \kcut in arbitrary graphs}

In this subsection we turn our attention to FLIP for \kcut in arbitrary graphs and prove Theorem \ref{theorem:max-k-cut-quasi-poly}.
We will again use Lemma \ref{lemma:ERprob}.
We note that it is not immediately clear that a result similar to Lemma \ref{lemma:rank-improved} holds for arbitrary graphs. 
The proof technique for Lemma \ref{lemma:rank-improved} fails for arbitrary graphs since the acyclic vertices between two cyclic blocks for a vertex $v$ may not be in the neighborhood of $v$.
Instead, we show that each sequence of sufficiently large length must have a block which has a large fraction of cyclic vertices.

\begin{definition}
Define the \emph{surplus} of a sequence $L$ as 
\[
z(L) := \ell(L) - \sum_{v \in A(L)} \#_L(v) - c(L)
\]
and the maximum surplus over all blocks of length $k$ in $L$ as $m_L(k) := \max \{z(B) : B \in [n]^k \}$.
\end{definition}

\begin{lemma}
\label{lem:blockexists}
Suppose $\alpha > 1$ and that $L$ is a sequence of length $\alpha n$. 
Then there exists a block $B$ in $L$ such that
\[
\frac{c(B)}{\ell(B)} \geq \frac{\alpha -k+1 }{(2k-1) \alpha \lg (\alpha n))}.
\]
\end{lemma}

\begin{proof}
We will use the following claim. 
\begin{claim} \label{claim}
If a block $B$ is the concatenation of blocks $B_1$ and $B_2$ then
\[
z(B) \leq z(B_1) + z(B_2) + (2k-1)c(B).
\]
\end{claim}

\begin{proof}
We recall Remark \ref{remark:acyclicmax} which states that an acyclic vertex appears at most $k-1$ times in a sequence.
Hence, in the worst case, a cyclic vertex $v \in C(B)$ in $B$ is an acyclic vertex in $B_1$ and $B_2$ and appears $k-1$ times in each block.
\end{proof}

Now suppose that there exists a $\delta$ such that for every block $B$ in $L$, we have $c(B) \leq \delta \ell(B)$. Then
\[
m_L(t) \leq 2 m_L(t/2) + (2k-1) \delta t
\]
for all $t > 0$.
Since $m(1) = 0$, we can bound the sequence above by
\begin{equation}
\label{eq:upb}
m_L(\alpha n) \leq (2k-1) \delta \alpha n \lg(\alpha n).
\end{equation}
Since,
\[
\sum_{v \in A(B)} \#_{B}(v) + c(B) \leq (k-1) s(B)
\]
for all blocks $B$ in $L$ and $s(L)\le n$, it follows that
\begin{equation}
\label{eq:lowb}
m_L(\alpha n) \geq \alpha n - (k-1)s(L) \geq (\alpha - k+1)n.
\end{equation}
Therefore, combining (\ref{eq:upb}) and (\ref{eq:lowb}), we get that
\[
\delta \geq \frac{\alpha -k+1 }{(2k-1) \alpha \lg (\alpha n)}
\]
which concludes the proof.
\end{proof}


%

\begin{definition}
A block $B$ is $\alpha$\emph{-cyclic} if
\[
c(B) \geq \frac{\alpha -k+1 }{(2k-1) \alpha \lg (\alpha n)} \ell(B).
\]
\end{definition}

\begin{lemma}\label{lemma:kgenbbound}
Let $G=(V,E)$ be a graph and let $B$ be an $\alpha$-cyclic block.
Then
\[
\Pr_X[\exists\ \tau_0\in [k]^V:\ B \text{ is $\epsilon$-slowly improving from $\tau_0$ with respect to $X$}] \leq k^{\ell (B)}
(k \phi \epsilon)^{ \frac{\alpha-k+1}{2(2k-1) \alpha \lg(\alpha n)} \ell(B)}.
\]
\end{lemma}

\begin{proof}
Suppose that $B$ is $\epsilon$-slowly improving from some $\tau_0 \in [k]^V$ with respect to $X$. Since $B$ is improving from $\tau_0$ with respect to $X$, it follows that $\langle M_{B,\tau_0}^t,X\rangle\in (0,\epsilon]$ for all $t\in [\ell(B)]$. Since every column of $P_{B,\tau_0}$ is the sum of at most $k$ columns of $M_{B,\tau_0}$, we have that $\langle P_{B,\tau_0}^C,X\rangle\in (0,k\epsilon]$ for all $C\in \Gamma(B)$. Hence, the required probability is at most
\[
\Pr_X\Big[\exists \tau_0\in [k]^V: \text{$B$ is valid from $\tau_0$ and $\langle P_{B,\tau_0}^C,X\rangle\in (0,k\epsilon]$ $\forall\ C\in \Gamma(B)$}\Big].\\
\]

Let $\mathcal{I}_{B,\tau_0,X}$ denote the event that $B$ is an improving sequence from $\tau_0$ with respect to $X$. Then, by union bound, the required probability is at most 
\[
\sum_{v\in S(B)}\sum_{\tau_0(v)\in [k]} \Pr_X\Big[\exists\ \tau_0(u)\in [k]\ \forall\ u\in V\setminus S(B):\ \mathcal{I}_{B,\tau_0,X} \text{ and } \langle P_{B,\tau_0}^C,X\rangle\in (0,k\epsilon]\ \forall\ C\in \Gamma(B) \Big].
\]

For $\pi_f:S(B)\rightarrow [k]$ and $\pi_c:V\setminus S(B)\rightarrow [k]$, let us define $\tau_{(\pi_f,\pi_c)}:V\rightarrow [k]$ as 
\[
\tau_{(\pi_f,\pi_c)}(u):=
\begin{cases}
\pi_f(u) & \text{ if $u\in S(B)$ and}\\
\pi_c(u) & \text{ if $u\in V\setminus S(B)$}.
\end{cases}
\]

Now, consider a fixed choice of $\pi_f$ and $\pi_c$. We would like to bound the following probability: 
\[
\Pr_X\Big[\exists\ \pi_c:V\setminus S(B)\rightarrow [k]:\ \mathcal{I}_{B,\tau_{(\pi_f,\pi_c)},X} \text{ and } \langle P_{B,\tau_{(\pi_f,\pi_c)}}^C,X\rangle\in (0,k\epsilon]\ \forall\ C\in \Gamma(B) \Big].
\]

Let us define an initial configuration 
\[
\sigma_0(u):=
\begin{cases}
\pi_f(u)&\text{ if $u\in S(B)$ and }\\
1&\text{ if $u\in V\setminus S(B)$}.
\end{cases}
\]
By Proposition \ref{prop:non-moving-nullified}, we have that $P_{B,\sigma_0}=P_{B,\tau_{(\pi_f, \pi_c)}}$ for every $\pi_c : V \setminus S(B) \rightarrow [k]$. 
Hence, 
\begin{align*}
&\Pr_X\Big[\exists\ \pi_c:V\setminus S(B)\rightarrow [k]:\ \mathcal{I}_{B,\tau_{(\pi_f,\pi_c)},X} \text{ and } \langle P_{B,\tau_{(\pi_f,\pi_c)}}^C,X\rangle\in (0,k\epsilon]\ \forall\ C\in \Gamma(B) \Big].\\
&\quad = \Pr_X\Big[\exists\ \pi_c:V\setminus S(B)\rightarrow [k]:\ \mathcal{I}_{B,\tau_{(\pi_f,\pi_c)},X} \text{ and } \langle P_{B,\sigma_0}^C,X\rangle\in (0,k\epsilon]\ \forall\ C\in \Gamma(B) \Big]\\
&\quad = \Pr_X\Big[\langle P_{B,\sigma_0}^C,X\rangle\in (0,k\epsilon]\ \forall\ C\in \Gamma(B) \big|\exists\ \pi_c:V\setminus S(B)\rightarrow [k]:\ \mathcal{I}_{B,\tau_{(\pi_f,\pi_c)},X}\Big]\\
&\quad \quad \quad \quad \times \Pr_X\Big[\exists\ \pi_c:V\setminus S(B)\rightarrow [k]:\ \mathcal{I}_{B,\tau_{(\pi_f,\pi_c)},X}\Big]\\
& \quad \le \Pr_X\Big[\langle P_{B,\sigma_0}^C,X\rangle\in (0,k\epsilon]\ \forall\ C\in \Gamma(B) \big|\exists\ \pi_c:V\setminus S(B)\rightarrow [k]:\ \mathcal{I}_{B,\tau_{(\pi_f,\pi_c)},X}\Big]. 
\end{align*}

Now, we bound the RHS probability. If there exists $\pi_c:V\setminus S(B)\rightarrow [k]$ such that the sequence $B$ is improving from $\tau_{(\pi_f,\pi_c)}$ with respect to $X$, then by Corollary \ref{coro:3crit}, the rank of $P_{B,\tau_{(\pi_f,\pi_c)}}$ is at least $c(B)/2$. Moreover, we know that $P_{B,\sigma_0}=P_{B,\tau_{(\pi_f,\pi_c)}}$ and hence the rank of $P_{B,\sigma_0}$ is at least $c(B)/2$. Since $B$ is an $\alpha$-cyclic block, we have that $c(B)\ge (\alpha-k+1)\ell(B)/((2k-1)\alpha\log{(\alpha n)})$ and hence the rank of $P_{B,\sigma_0}$ is at least $(\alpha-k+1)\ell(B)/(2(2k-1)\alpha\log{(\alpha n)})$. 
Therefore, using Lemma \ref{lemma:ERprob}, the RHS probability is at most 
\[
(k\phi\epsilon)^{\text{rank}\left(P_{B,\sigma_0}\right)}\le (k\phi \epsilon)^{\frac{1}{2}\frac{\alpha - k+1}{(2k-1) \alpha \lg(\alpha n)}\ell(B)}.
\]

Hence, the probability required in the lemma is at most
\[
k^{s(B)}(k\phi \epsilon)^{\frac{1}{2} \cdot \frac{\alpha - k+1}{(2k-1) \alpha \lg(\alpha n)}\ell(B)}.
\]
as the number of possible initial configurations for the vertices that move in $B$ is at most $k^{s(B)}$. The lemma now follows since $s(B)\le \ell(B)$ and $k\ge 1$.

\end{proof}


\begin{lemma}\label{lemma:kepsimp}
Let $G$ be a graph and let $\epsilon=\phi^{-1} n^{-((2k-1)k \lg(kn)+ \eta)}$ for a constant $\eta>0$. 
Then, the probability (over the choices of $X$) that there exists a sequence $L$ of moves of length $kn$ and an initial configuration $\tau_0\in [k]^V$ such that $L$ is $\epsilon$-slowly improving from $\tau_0$ with respect to $X$ is $o(1)$. 
\end{lemma}
\begin{proof}
Let $\mathcal{R}_X$ denote the event that there exists a sequence $L$ of moves of length $kn$ and an initial configuration $\tau_0\in [k]^V$ such that $L$ is $\epsilon$-slowly improving from $\tau_0$ with respect to $X$.
It follows from Lemma \ref{lem:blockexists} that every sequence $L$ of length $kn$ contains a $k$-cyclic block.
Therefore, if the event $\mathcal{R}_X$ happens, then there exists a $k$-cyclic block $B$ of length at most $kn$ and an initial configuration $\tau_0\in [k]^V$ such that $B$ is $\epsilon$-slowly improving from $\tau_0$ with respect to $X$.
Hence, 
\begin{align*}
\Pr_X[\mathcal{R}_X] 
&\le \sum_{\substack{B:\ B\text{ is $k$-cyclic, }\\ \ell(B)\leq kn}}\Pr_X\Big[\exists \sigma_0 \in [k]^V:\ B \text{ is $\epsilon$-slowly improving from $\tau_0$ with respect to $X$}\Big]\\
& \leq \sum_{\substack{B:\ B\text{ is $k$-cyclic, }\\ \ell(B)\leq kn}} k^{\ell(B)} {{(k\phi \epsilon)^{\frac{1}{2(2k-1)k \lg(kn)}\ell(B)}}}
\quad \quad \quad \text{(by Lemma \ref{lemma:kgenbbound})}
\\
& \leq \sum_{\ell =1}^{kn} n^{\ell} k^{\ell} {(k\phi \epsilon)^{\frac{\ell}{2(2k-1)k \lg(kn)}}} \\
& \leq \sum_{\ell =1}^{kn} \left(k^2 n \phi^{\frac{1}{2(2k-1)k \lg(kn)}}\epsilon^{\frac{1}{2(2k-1)k \lg(kn)}}\right)^\ell
\end{align*}
where the third inequality follows from the fact that there are at most $n^\ell$ blocks of length $\ell$.
Therefore, for constant $k$ and $\epsilon=\phi^{-1} n^{-(2(2k-1)k \lg(kn)+ \eta)}$, the sum tends to 0 as $n \rightarrow \infty$.

\end{proof}

The proof of Theorem \ref{theorem:max-k-cut-quasi-poly} follows from Lemma \ref{lemma:kepsimp} similar to the proof of Theorem \ref{theorem:max-cut-poly} that follows from Lemma \ref{lemma:epsimp-2}.
In this case, we consider the event that an implementation of FLIP produces a sequence of moves from some initial configuration which has length greater than $\phi n^{2(2k-1)k \lg(kn)+3+\eta}$ for any constant $\eta > 0$.
This event implies that there exists a 
sequence of length $kn$ that is 
$\epsilon$-slowly improving from some initial configuration, where $\epsilon  = \phi^{-1}n^{-(2(2k-1)k \lg(kn)+ \eta))}$. 
Finally, we note that the probability of such an event is $o(1)$ by Lemma \ref{lemma:kepsimp}.


\paragraph{Acknowledgements.} We would like to thank Alexandra Kolla for bringing this problem to our attention. 
\bibliographystyle{alpha}
\bibliography{references}

\appendix

\section{Proof of Lemma \ref{lem:prob}}	
\label{app:prob}

We restate and prove Lemma \ref{lem:prob}.
\lemmakbbound*

\begin{proof}
Our proof closely resembles that of Lemma A.1 in \cite{ER17}.
Let $e_i$ denote the $i$-th coordinate vector.
We can extend $\{\alpha_1, \dots, \alpha_k\}$ to a basis for $\R^m$ by adding
coordinate vectors.
Without loss of generality suppose that the derived basis is $\{\alpha_1, \dots, \alpha_k, e_{k+1}, \dots, e_m\}$.
Let 
\[
B:=\begin{bmatrix}\alpha_1^\intercal \\
\vdots \\
\alpha_k^\intercal \\
e_{k+1}^\intercal \\
\vdots \\
e_m^\intercal
\end{bmatrix}.
\]
We note that $A$ is a full-rank $m \times m$ matrix.
Let $Y:=(Y_1,\dots,Y_m)^T:=BX$ and $A:=B^{-1}$.
We note that $Y$ is a random vector whose coordinates are possibly dependent.
Let $f,g:\R^m \rightarrow \R_{\geq 0}$ denote the joint density function of $X$ and $Y$ respectively.
\begin{claim}
Let $y_1,\dots,y_m \in \R$. Then
$g(y_1,\dots,y_m) \leq \phi^k \prod_{i=k+1}^m f_i(y_i)$.
\end{claim}

\begin{proof}
Let $y:=(y_1,\dots,y_m)$.
Then, we have 
$g(y) = |\det(A)^{-1}|f(A^{-1}(y))$.
Since $A$ is integral, we have that $|\det(A)| \in \Z_+$ and hence $|\det(A)^{-1}| \leq 1$.
Therefore,
\begin{align*}
g(y) &\leq f(A^{-1}(y))\\
&=f(A_1^{-1}(y), \dots, A_m^{-1}(y)) \\
&=f_1(A_1^{-1}(y)) \dots f_m(A_m^{-1}(y)) \quad \quad \quad \text{(since $X_1,\dots,X_m$ are independent)}\\
&\leq \phi^k f_{k+1}(A_{k+1}^{-1}(y)) \dots f_m(A_{m}^{-1}(y)) \\
&= \phi^k f_{k+1}(B_{k+1}(y)) \dots f_m(B_m(y)) \\
&= \phi^k f_{k+1}(y_{k+1}) \dots f_m(y_m).
\end{align*}
\end{proof}

Let $V:=\{(y_1,\dots,y_m) \in \R^m: y_1,\dots,y_k > 0, \sum_{i=1}^{k} y_i \leq \epsilon\}$
and $U:=\{(y_1,\dots,y_k) \in \R^k: y_1,\dots,y_k > 0, \sum_{i=1}^k y_i \leq \epsilon\}$. We have that
\begin{align*}
\Pr_X \Big[\alpha_i^\intercal X > 0 \ \forall i \in [k], \ \sum_{i=1}^k \alpha_i^\intercal X \leq \epsilon \Big]
&= \Pr_Y \Big[y_i > 0 \ \forall i \in [k], \ \sum_{i=1}^k Y_i \leq \epsilon \Big]\\
&= \int_V g(y_1,\dots,y_m)dy_1\dots dy_m \\
&\leq \int_V \phi^k \left(\prod_{i=k+1}^m f_i(y_i)\right) dy_1\dots dy_m \\
&= \phi^k \left( \int_U dy_1\dots dy_k \right)\left(\int_{y_{k+1},\dots,y_m \in \R} \prod_{i=k+1}^m f_i(y_i) dy_{k+1} \dots dy_m\right)\\
&\leq \phi^k \text{Vol}(U)\left(\prod_{i=k+1}^m \left(\int_{y_i \in \R} f_i(y_i)dy_i\right)\right)\\
&= \phi^k \text{Vol}(U) \\
&= {\phi^k}\frac{\epsilon^k}{k!}. \quad \quad \quad \text{(as $U$ is a $k$-dimensional simplex of side length $\epsilon$)}
\end{align*}

\end{proof}

\end{document}